\documentclass[11pt]{amsart}

\usepackage{fullpage}

\usepackage{ graphicx, amsmath, amssymb}

\usepackage{tikz}
\usepackage{tikz-dependency}

\tikzset{  diffVertex/.style={circle, draw=black, thick, fill= white,  minimum size =2.5mm, inner sep=0mm},
	diffVertexGray/.style={circle, draw=gray, thick, fill= white,  minimum size =2.5mm, inner sep=0mm},
	intVertex/.style={circle, draw=black, thick, fill= black,  minimum size =2.5mm, inner sep=0mm},
	diffTree/.style={regular polygon, regular polygon sides=4, , draw=black, thick, fill= white,  minimum size =3.2mm, inner sep=0mm}, 
	intTree/.style={regular polygon, regular polygon sides=4, , draw=black, thick, fill=black,  minimum size =3.2mm, inner sep=0mm}, 
	-|-/.style={decoration={markings, 	mark=at position .5 with {\arrow{|}}},postaction={decorate}}
}

\usetikzlibrary{arrows}	
\usetikzlibrary{decorations}
\usetikzlibrary{decorations.markings}

\usepackage{hyperref}
\usepackage[nameinlink]{cleveref}

\usepackage{placeins} 

\newtheorem{theorem}{Theorem}[section]
\newtheorem{lemma}[theorem]{Lemma}

\theoremstyle{definition}
\newtheorem{definition}[theorem]{Definition}
\newtheorem{example}[theorem]{Example}

\numberwithin{equation}{section}

\usepackage[
backend=biber,
]{biblatex}
\newcommand*{\mb}{\mathbb}
\newcommand*{\mc}{\mathcal}
\newcommand{\abs}[1]{\left | #1 \right |}
\renewcommand{\d}{\textnormal{d}}

\title{Perturbation theory of transformed quantum fields}
\author{Paul-Hermann Balduf}

\thanks{The author thanks Dirk Kreimer and Karen Yeats for helpful discussion.}

\begin{document}
	
\maketitle

\begin{abstract}
 We consider a scalar quantum field $\phi$ with arbitrary polynomial self-interaction in perturbation theory.
 If the field variable $\phi$ is repaced by a local diffeomorphism $\phi(x) = \rho(x) + a_1 \rho^2(x) +\ldots$, this field $\rho$ obtains infinitely many additional interaction vertices. We show that the $S$-matrix of $\rho$ coincides with the one of $\phi$ without using path-integral arguments. This result holds even if the underlying field has a  propagator of higher than quadratic order in the momentum.
 
 If tadpole diagrams vanish, the diffeomorphism can be tuned to cancel all contributions of an underlying $\phi^s$-type self interaction at one fixed external offshell momentum, rendering $\rho$ a free theory at this momentum. 
 
 Finally, we propose one way to extend the diffeomorphism to a non-local transformation involving derivatives without spoiling the combinatoric structure of the local diffeomorphism.
\end{abstract}

\section{Introduction}

\subsection{Motivation and content}

A quantum field theory can be defined via a Lagrangian density $\mc L$, simply called Lagrangian hereafter.  In perturbative computations in this theory, the monomials of degree $n>2$ in $\mc L$ correspond  to $n$-valent interaction vertices in  Feynman diagrams. The monomials of degree two define the propagator of the field. We only consider scalar quantum fields in this paper. A free scalar quantum field theory has no self-interaction and is defined via the Lagrangian 
\begin{align}\label{Lfree} 
\mc L_\phi  &=\frac 12   \partial_\mu \phi(x) \partial^\mu \phi(x) -\frac 12 m^2 \phi^2(x),
\end{align}
where $m\geq 0 $ is the mass of a $\phi$-particle. We allow $m=0$ for the massles theory.  $x$ is a point in spacetime which we take to be 4-dimensional for concreteness even if the results do not depend on this choice.

We express the field $\phi(x)$ as a diffeomorphism
\begin{align}\label{def_diffeomorphism} 
\phi(x) &= \sum_{j=0}^\infty a_j \rho^{j+1}(x), \qquad a_0=1
\end{align}
of another field $\rho(x)$.  $\{a_j\}_{j\geq 1}$ are constants with respect to spacetime. The constraint $a_0=1$ means the diffeomorphism is tangent to identity, i.e. the fields $\rho$ and $\phi$ coincide at leading order.  This replacement, when applied to the Lagrangian density $\mc L_\phi$ of $\phi$, gives rise to a Lagrangian  $\mc L_\rho$ of $\rho$ which generally involves monomials of any order. Thus,  $\rho$ is an interacting quantum field theory even if the original field $\phi$ was free.

  If \cref{def_diffeomorphism} is applied to the classical field theory before canonical quantization, it amounts to a canonical transformation which leaves the poisson brackets intact  \cite{nakanishi_covariant_1990} and only changes the Lagrangian, thus relating theories with different Lagrangians.
  
In the framework of the path integral,  observables appear invariant under diffeomorphisms as the diffeomorphism can be undone by a redefinition of the integration measure. However, different lines of argument lead to different results \cite{apfeldorf} which might be due to operator ordering ambiguities \cite{pointPathIntegral}. One possibility to resolve these puzzles is to understand field diffeomorphisms order by order in perturbation theory and then, in a later step, extend these results to a statement about the generating functionals of Feynman graphs. This might be a mathematically cleaner way than direct, formal manipulations of field variables within (divergent) generating functionals. This argument is inspired by the recent proof \cite{jackson_robust_2017} that the Legendre transform - which relates the generating functionals of connected Feynman graphs to that of 1PI-graphs - can be understood order by order without  problems regarding the divergence of said generating functionals in quantum field theory. 
In this paper we focus on the first of the two steps, i.e. the change of the Lagrangian density $\mc L_\phi \mapsto \mc L_\rho=\mc L_\phi(\phi(\rho))$ which is induced by a transformation $\phi \mapsto \rho$ and examination of the Feynman rules and correlation functions of the theory defined by $\mc L_\rho$. 

This paper is a continuation of earlier work by Kreimer, Yeats and Velenich \cite{KY17,velenich}.
In the remainder of this section and \cref{sec_free}, we will set up notation and review the results of \cite{KY17}, which mainly consider a diffeomorphism of an underlying free theory. Many concepts developed there can be readily applied to diffeomorphisms of an underlying interacting theory, which will be done in \cref{subs_int}.   Next, we proceed  in \cref{sec_cancellation} to a possible application. Namely, we show that a diffeomorphism can - for offshell correlation functions - alter the type of interaction present in the theory. In \cref{sec_arbitrary} we note that the results of \cite{KY17} for the free theory are not limited to the specific momentum dependece of the propagator used there.  Finally, in \cref{subs_nonlocal} we present a possibility of including derivatives into the transformation $\phi \mapsto \rho$ such that the welcome combinatorial structure of the local diffeomorphism is conserved.  
 
\subsection{Prerequisites}

A Lagrangian  is a power series in the field variable $\phi$, hence replacing $\phi$ by $\phi(\rho)$ according to \cref{def_diffeomorphism} amounts to an insertion of power series into each other. For formal power series, the coefficient of the concatenation are given by Fa\`{a} di Bruno's formula \cite{flajolet_analytic_2009}: If
\begin{align*} 
f(t) &= \sum_{n=1}^\infty f_n t^n \quad \text{ and } \quad  g(t) = \sum_{n=0}^\infty g_n t^n 
\end{align*} 
then
\begin{align}\label{faadibruno}
[t^n] \left( f\left( g(t) \right)   \right)  &=\frac{1}{n!} \sum_{k=1}^n k!f_k \cdot B_{n,k} \left( 1! g_1, 2! g_2 \ldots,(n+1-k)! g_{n+1-k} \right) .
\end{align}
Here, $[t^n]$ denotes extraction of the $n^{\text{th}}$ coefficient (i.e. $[t^n]f(t) = f_n$) and  $B_{n,k}$ are the partial Bell polynomials, defined via
\begin{align}\label{bell_generating}
\sum_{n=0}^\infty \sum_{k=0}^n B_{n,k} \left( x_1, x_2, \ldots \right) u^k \frac{t^n}{n!} &= \exp \left( u \sum_{j=1}^\infty x_j \frac{t^j}{j!}\right) . 
\end{align}
Bell Polynomials count the possible partitions of  $\left \lbrace 1, \ldots, n \right \rbrace $ into  $k$ nonempty disjoint subsets:
\begin{align*}
B_{n,k} \left( x_1, x_2, x_3, \ldots \right) &=\sum_P x_{\abs {P_1}} \cdots x_{\abs{P_k}}
\end{align*}
where
\begin{align}\label{bell_partitions}
P &= \big \lbrace  \emptyset \neq P_i \subseteq \left \lbrace 1, \ldots, n \right \rbrace ~\forall i, \quad  P_i \cap P_j = \emptyset ~ \forall i\neq j, \\
&\qquad P_1 \cup \ldots \cup P_k = \left \lbrace 1, \ldots, n \right \rbrace  \big \rbrace. \nonumber
\end{align}

Inversion of a power series is given by Lagrange inversion \cite{lagrangeInversion}, which again  applies for formal power series regardless of convergence as functions \cite{henrici_algebraic_1964}:
\begin{align}\label{lagrange_inversion} 
 \left( f^{-1} \right) _n &=   \frac 1 {n!} \sum_{k=1}^{n-1}   \frac{1}{f_1^{n+k} }B_{n-1+k, k} \left( 0, -2! f_2,- 3! f_3 , \ldots  \right), \quad \left( f^{-1} \right) _1=\frac{1}{f_1} . 
\end{align}
We note in passing that  concatenation and inversion of formal power series can also be interpreted as coproduct and antipode of the Fa\`{a} di Bruno Hopf algebra \cite{figueroa_faa_2005}.

 For fixed $a\in \mb N_0$ and $b\in \mb N$ the \emph{Fuss-Catalan numbers} are defined as
	\begin{align}\label{fc}
	F_m (a,b) &= \frac{b}{ma+b}\binom{ma+b}{m} = b\frac{(ma+b-1)! }{(ma+b-m)! m!}.
	\end{align}

We will also use the Euler characteristic which relates the number of vertices $\abs{V_\Gamma}$, (internal) edges $\abs{E_\Gamma}$ and loops $\abs \Gamma$ of a connected graph $\Gamma$: 
\begin{align}\label{euler_graph}
\abs{V_\Gamma} - \abs{E_\Gamma} + \abs{\Gamma} &=1.
\end{align}

\section{Local diffeomorphisms of a free field}\label{sec_free}
A significant part of this section is  a review of  \cite{KY17}, introducing concepts and results needed for the current paper.

\subsection{Feynman rules }

\begin{definition}\label{def_offshellvariable} 
	For a 4-momentum $p$ in a theory with mass $m\geq 0$, the corresponding \emph{offshell variable} is defined as
	\begin{align*}
	x_p &:= p^2-m^2.
	\end{align*}
\end{definition}
This generalizes to sums of numbered momenta in a slight abuse of notation, e.g. $x_{2+5} :=x_{p_2 + p_5}\equiv  (p_2+p_5)^2-m^2$. We will also use the notation $x_e:=p_e^2-m^2$ where $e$ is some edge in a graph with momentum $p_e$ flowing through it.

Applying the diffeomorphism \cref{def_diffeomorphism} to the free Lagrangian  \cref{Lfree}  and collecting the derivatives of the kinetic term by partial integration yields
\begin{align}\label{L_free_rho}  
\mc L_\rho := \mc L_\phi(\phi(\rho))&= - \sum_{n=1}^{\infty}  \frac{f_{n+1}}{n!} \rho^{n} \partial_\mu \partial^\mu \rho    -m^2\sum_{n=2}^{\infty} \frac{  c_{n-2}}{n!} \rho^{n} .
\end{align}
The quantities $f_n, g_n$ appear as coupling constants induced by the diffeomorphism. 
\begin{align*}
f_n&= B_{n-2,1} (  2! a_1, 3! a_2, \ldots) + B_{n-2,2} (2! a_1, 3! a_2, \ldots ),  \\
 c_{n-2}&=B_{n,2} \left( 1, 2! a_1,  \ldots  \right) ,\\
g_n &:= n f_n - c_{n-2}= \frac {n(n-2)!}{2 } \sum_{k=0}^{n-2} a_{n-k-2} a_k   (n-k-2) k .
\end{align*}
From \cref{L_free_rho}, one reads off the $n$-valent  vertex Feynman rules
\begin{align}\label{diff_vn}
iv_n &=  i f_n \cdot \left( x_1 + \ldots + x_n \right)  + i g_n m^2, \quad n \geq 2. 
\end{align}
We will subsequently call these vertices \emph{diffeomorphism vertices}.
The first ones together with their explicit Feynman rules are depicted in \cref{free_vertex_picture}.

 \begin{figure}[htbp]
	\centering
	\begin{tikzpicture}

	\node at (-1,0) {$iv_3=$};
	
	\node [diffVertex]  (c) at (0,0) {} ;
	
	\draw [thick] (c) -- + (90:.5);
	\draw [thick] (c) -- + (210:.5);
	\draw [thick] (c) -- + (330:.5);
	
	\node at (.5,0) [anchor=west,text width = 9 cm] {$=2ia_1 \left(x_1+x_2+x_3\right)$};

	\node at (-1,-1.5) {$iv_4=$};
	
	\node [diffVertex]  (c) at (0,-1.5) {} ;
	
	\draw [thick] (c) -- + (45:.5);
	\draw [thick] (c) -- + (135:.5);
	\draw [thick] (c) -- + (225:.5);
	\draw [thick] (c) -- + (315:.5);
	
	\node at (.5,-1.5)  [anchor=west,text width = 9cm]{$=  4im^2 a_1^2 + i\left( 6 a_2+4 a_1^2 \right)  \left(x_1+x_2+x_3+x_4\right) $};

	\node at (-1,-3) {$iv_5=$};
	
	\node [diffVertex]  (c) at (0,-3) {} ;
	
	\draw [thick] (c) -- + (90:.5);
	\draw [thick] (c) -- + (162:.5);
	\draw [thick] (c) -- + (234:.5);
	\draw [thick] (c) -- + (306:.5);
	\draw [thick] (c) -- + (18:.5);
	
	\node at (.5,-3) [anchor=west,text width = 11cm]{$= 60 im^2 a_1 a_2 +  i\left( 24 a_3+36 a_1 a_2  \right)  \left(x_1+x_2+x_3+x_4+x_5\right) $};

	\end{tikzpicture}
	\caption{Graphical representation of the diffeomorphism vertices $iv_n$ from \cref{diff_vn}. The numbered momenta correspond to edges adjacent to the vertex.}
	\label{free_vertex_picture}
\end{figure}

\FloatBarrier
\subsection{Tree sums with one external edge offshell}

Since the diffeomorphism \cref{def_diffeomorphism} is tangent to identity, the 2-point vertex is unaltered. The field $\rho$ has the same propagator (=non-amputated time ordered 2-point function) as $\phi$, with \cref{def_offshellvariable}
\begin{align}\label{free_propagator} 
 \Gamma_{2, \text{free}}^{-1} = \langle \rho(p) \rho(-p)\rangle  &= \frac{i}{p^2-m^2}=\frac{i}{x_p}.
\end{align}

\begin{definition}\label{def_bn}
	The tree sums $b_n$ for $n\geq 2$ are defined as the sum of all trees with $n$ external edges onshell (i.e. $x_j=0$ for these edges $j$) and additionally one external edge offshell, where the propagator of this offshell edge is included in $b_n$. Further, $b_1:=1$.
\end{definition}

The construction of $b_n$ is illustrated in \cref{free_b1_b2_b3_picture}.

\begin{example}\label{ex_bn}
	An explicit calculation using \cref{diff_vn} yields
	\begin{align*} 
	b_2 &= \frac{i}{x_{1+2}}\cdot 2ia_1 \left( x_1 + x_2 + x_{1+2} \right) \Big|_{x_1=0=x_2} = -2 a_1\\
	b_3 &= -6 a_2 + 12 a_1^2.
	\end{align*}
\end{example}

\begin{figure}[h]
	\centering
	\begin{tikzpicture}

	\node at (-2,1.5)  [anchor=west,text width = 1cm]{$b_1=$};
	
	\node [diffTree] (c) at (-.8,1.5) {};
	\draw [thick] (c) -- + (90:.5);
	\draw [-|,thick] (c) -- +(270:.4);
	
	\node at (-.2,1.5) {$=$};
	
	\node at (.5, 1.5){$1$};

	\node at (-2,0) [anchor=west,text width = 1cm] {$b_2=$};
	
	\node [diffTree]  (c) at (-.8,0) {} ;
	
	\draw [thick] (c) -- + (90:.5);
	\draw [ thick, -|,bend angle=20,bend left] (c.245) to +(240:.4);
	\draw [thick, -|, bend angle=20, bend right] (c.295) to +(300:.4); 
	
	\node at (-.2 ,0) {$=$};
	
	\node [diffVertex]  (c) at (.6,0) {} ;
	
	\draw [thick] (c) -- + (90:.5);
	\draw [-|,thick] (c) -- +(240:.5);
	\draw [-|,thick] (c) -- +(300:.5); 
	
	\node (of) at (2, .8) [anchor=west,text width = 7cm] {offshell edge, propagator included};
	\node (on) at (2, 0) [anchor=west,text width = 7cm] {onshell edges, no propagator};
	
	\draw [->, bend angle=10, bend right] (of.180) to (.8,.4);
	\draw [->,  bend angle=10, bend right] (on.180) to (.9,-.2);
	
	\node at (-2,-1.5) [anchor=west,text width = 1cm] {$b_3=$};

	\node [diffTree]  (c) at (-.8,-1.5) {} ;
	
	\draw [thick] (c) -- + (90:.5);
	\draw [-|, thick, bend angle=20, bend left] (c.240) to +(230:.5);
	\draw [-|, thick] (c.270) -- +(270:.5); 
	\draw [-|, thick, bend angle=20, bend right] (c.300) to +(310:.5); 
	
	\node at (-.2,-1.5) {$=$};
	
	\node [diffVertex] (c1) at (.6,-1.5) {};
	\draw [thick] (c1) -- + (90:.5);
	\draw [-|,thick] (c1) -- +(220:.7);
	\draw [-|,thick] (c1) -- +(270:.6); 
	\draw [-|,thick] (c1) -- +(320:.7);
	
	\node at (1.8,-1.6)   {$+$};
	\node [diffVertex] (c1) at (3,-1.3) {};
	\node [diffVertex] (c2) at (3.3,-1.7) {};
	\draw [thick] (c1) -- (c2);
	\draw [thick] (c1) -- + (90:.5);
	\draw [-|,thick] (c1) -- +(230:.8);
	\draw [-|,thick] (c2) -- +(240:.5); 
	\draw [-|,thick] (c2) -- +(310:.5);
	
	\node at (4,-1.6)   {$+$};
	\node [diffVertex] (c1) at (5,-1.3) {};
	\node [diffVertex] (c2) at (4.7,-1.7) {};
	\draw [thick] (c1) -- (c2);
	\draw [thick] (c1) -- + (90:.5);
	\draw [-|,thick] (c1) -- +(310:.8);
	\draw [-|,thick] (c2) -- +(240:.5); 
	\draw [-|,thick] (c2) -- +(310:.5);
	
	\node at (6,-1.6)   {$+$};
	\node [diffVertex] (c1) at (6.9,-1.3) {};
	\node [diffVertex] (c2) at (7.2,-1.7) {};
	\draw [thick] (c1) -- (c2);
	\draw [thick] (c1) -- + (80:.5);
	\draw [-|,thick] (c1) -- +(270:.9);
	\draw [-|,thick] (c2) -- +(205:1); 
	\draw [-|,thick] (c2) -- +(310:.5);

	\end{tikzpicture}
	\caption{
		The first tree sums $b_n$. The onshell edges are indicated with a perpendicular line at the end.   
	}
	\label{free_b1_b2_b3_picture}
\end{figure}
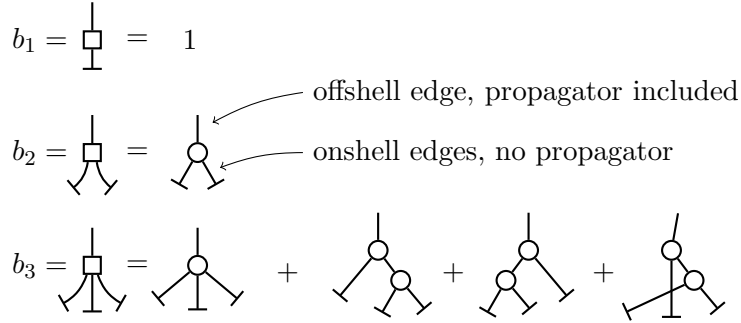

The tree sums $_n$ computed in \cref{ex_bn} are independent of masses and momenta. This continues even for higher $n$, the following remarkable result was shown in 	\cite[Thm.~3.5]{KY17}:
\begin{theorem}\label{thm_bn}
	\begin{align*}
	b_{n+1} &= \sum_{k=1}^{n}\frac{(n+k)!}{n!}B_{n,k} \left( -1! a_1, -2! a_2, \ldots, -n! a_n\right) .
	\end{align*}
\end{theorem}

\begin{definition}\label{def_Ajn}
	$A^j_n$ is the sum of all trees with a total of $n$ external edges, at most $j$ of which are offshell. 
\end{definition}
This definition implies that amplitudes with more than $j$ external edges offshell include $A^j_n$ as a summand or, in the other direction, amplitudes with less than $j$ external edges offshell can be extracted from $A^j_n$ by setting some of the $j$ edges onshell and then symmetrizing, see \cref{ex_A4}.
Unlike $b_n$, the tree sum $A^j_n$ does not include the propagators of external offshell edges.  

By \cref{thm_bn}, the $b_n$ are independent of masses and momenta. Since they include the propagator of the single offshell edge $\frac{i}{x_{1+2+\ldots + n}}$, the symmetric amputated  tree sum with $n+1$ external edges, at most one of which is offshell, is given by
\begin{align}\label{free_A1n} 
A^1_{n+1} :=-i\left( x_1 + x_2 + \ldots + x_n +  x_{1+2+\ldots + n}\right) b_n. 
\end{align}
In terms of graphs, this implies that the sum of all trees with a given number of external edges, one of which is offshell, effectively is a vertex with amplitude $A^1_{n+1}$.

\begin{lemma}\label{lem_Sn_free}
	If $A^0_n$ is the connected tree level Feynman amplitude with $n>2$ external edges, all of which are onshell from \cref{def_Ajn}, then
	\begin{align*} 
	A^0_n &=0.
	\end{align*}
\end{lemma}
\begin{proof}
	Set all $x_j=0$ in \cref{free_A1n} to obtain the amplitude $A^0_{n+1}$ where none of the edges is offshell.
\end{proof}
A direct consequence of \cref{lem_Sn_free} is that if $j>0$ in \cref{def_Ajn} then at least one of the edges actually \emph{is} offshell as the summand $A^0_n$ contained in $A^j_n$ is zero.

Apart from their interpretation as Feynman amplitudes of tree sums, the $b_n$ also have an equally remarkable second interpretation: 
\begin{lemma}\label{lem_inverse}
	The tree sums $b_n$ are  coefficients of the inverse diffeomorphism $\rho(\phi)$ ,
	\begin{align*} 
	\rho(x) &=   \sum_{n=1}^\infty  \frac{b_n}{n!} \phi^n(x).
	\end{align*}
\end{lemma}
\begin{proof}
	Compute the coefficients of the inverse diffeomorphism using Lagrange inversion \cref{lagrange_inversion}. Extracting coefficients from \cref{bell_generating}, one confirms
	\begin{align*} 
	B_{n+k,k} \left( 0,  -2! a_1, -3! a_2, \ldots \right)  &= \frac{(n+k)!}{n!}B_{n,k} \left( -1! a_1, -2! a_2, -3! a_3, \ldots \right).
	\end{align*}
\end{proof}

\subsection{Uncancelled edges as cuts}\label{subs_cuts}

The Feynman rules \cref{diff_vn} of the diffeomorphism vertices $iv_n$ are -- up to summands $m^2$ -- proportional to offshell variables $x_e$ of adjacent edges $e$. In a Feynman diagram, such edges come with a propagator \cref{free_propagator} with Feynman amplitude $\frac{i}{x_e}$. That is, the propagator together with the corresponding summand of the vertex evaluates to a constant. This is the mechanism underlying \cref{thm_bn}, the vertex effectively \emph{cancels} the edge $e$. Clearly, the possibility to cancel the edge depends on the presence of the summand $\propto x_e$ in the vertex Feynman rule, on the other hand, setting $x_e=0$ corresponds to the momentum flow through the edge $e$ being onshell.  The interplay between the notions of \emph{onshellness} and \emph{cancellation} is crucial for the subsequent discussion of tree level amplitudes.

Consider any   edge $e$ which connects two different tree sums $A^{j_1}_{n_1}, A^{j_2}_{n_2}$. Let $e$ have the offshell variable $x_e= p_e^2-m^2$. We observe the following:
\begin{enumerate}
	\item If the edge is cancelled by one of the adjacent tree sums, this tree sum includes a summand proportional to $x_e$. Thus, the edge $e$ is an offshell external edge from the point of view of this tree sum. 
	\item If the edge $e$ is not cancelled, then none of the two adjacent tree sums must contain a factor $x_e$. This can be obtained by formally setting $x_e=0$ in these tree sums, i.e. $e$ becomes an onshell external edge of the tree sums. 
\end{enumerate}
Setting $x_e=0$ for the factors does not mean that $x_e$ actually is zero for the whole, connected diagram, but just inside the two individual constituents. This is a purely formal procedure to eliminate the factors $x_e$ in the amplitudes $A^j_n$, consisting of two steps: First, setting $x_e=0$ removes the unwanted contributions, second, the remaining amplitude is interpreted as if the edge momentum $p_e$ were completely arbitrary.  The overall diagram where $e$ is uncancelled has an amplitude
\begin{align*} 
A^{j_1}_{n_1} \Big|_{x_e=0} \cdot \frac{i}{x_e}\cdot A^{j_2}_{n_2}\Big|_{x_e=0}.
\end{align*}
We noted below \cref{def_Ajn} that setting one of the $j_1$ external offshell edges onshell turns $A^{j_1}_{n_1}$ into $A^{j_1-1}_{n_1}$. This is not what we are doing here: In the present case, we set one \emph{specific} edge $e$ onshell whereas in the above case it was an arbitrary edge and the result is to be symmetrized. Especially, $A^{j_1}_{n_1} \big|_{x_e=0}$ is no longer symmetric in its $n$ external edges.

One might think that due to the momentum dependence of $A^j_n$, setting $x_e=0$ for an external edge $e$ does more than just eliminating the summands proportional to a power of $x_e$. This is  not the case as can be understood from the general structure of the Feynman rules:  Thanks to the absence of 2-valent vertices, no internal propagator in a tree has the Feynman rule $\frac{i}{x_e}$ and external propagators are not included into $A^j_n$. Hence there is no factor $\frac{1}{x_e}$ in $A^j_n$. There are of course factors from internal propagators of the form $\frac{1}{x_{e+f+\ldots}}$, where $p_f$ is some other momentum.  But those are unaffected, for example
\begin{align*} 
x_{e+f} \Big|_{x_e=0}  = \Big( \left( p_e + p_f \right) ^2-m^2  \Big)_{p_e^2=m^2}  = m^2 + 2 p_e p_f + p_f^2 -m^2 \neq  p_f^2-m^2=x_f.
\end{align*}
That is, the offshell variables $x_{e+f+\ldots}$ in the amplitude are not determined even if one (or more) of the involved momenta $p_e, p_f, \ldots$ are onshell.  Fixing the magnitude of the vectors imposes some constraint, e.g. the value of $p_e p_f$ is bounded by $m^2$ if both vectors have magnitude $m$, but the only thing important to us is that $x_{e+f}$ remains as an undetermined variable. If we in the second step remove the constraint $x_e=0$ then $x_{e+f}$ can again take completely arbitrary values but the factors $x_e$ in the numerator are gone. This way, it is possible to eliminate all numerator factors proportional to $x_e$ without otherwise altering the amplitude. See \cref{ex_A4,ex_bsn,ex_As4} for  explicit results, setting external edges onshell there does not symbolically alter the remainder of the amplitudes apart from eliminating terms proportional to the corresponding $x_e$.

If $\Gamma$ is a tree graph, we can consider any uncancelled internal edge $e$ as a cut: The resulting amplitude is the product of amplitudes of the individual components (this is always true for a tree graph), but these components have the same Feynman rule as if $e$ were an onshell external edge.  This phenomenon motivated the use of analytic properties of the $S$-matrix in \cite{velenich}.

\subsection{Offshell tree sums }\label{subs_offshell}
It is possible to obtain tree sums with more than one external edge offshell (\cref{def_Ajn}) from the $b_n$ as indicated in \cite[Sec.~4.2.1]{KY17}. To this end, the external onshell edges of multiple $b_k$ are glued. This leaves an uncancelled internal propagator, which, in the sum over all trees, becomes a symmetric sum of all partitions of the external edges. Especially, for $j>1$ the quantity $A^j_n$ depends on external masses and momenta.

\begin{example}\label{ex_A4}
	The amplitude with four external edges, all of which are possibly offshell, is
	\begin{align*} 
	A_4^4 	     &=   -ib_3 \left( x_1+x_2+x_3+x_4 \right)     -i b_2^2\left( \frac{(x_1+x_2)(x_3+x_4)}{x_{1+2}} +\text{ 2 more}  \right) .
	\end{align*}
	Here, $b_j$ are the free diffeomorphism tree sums (\cref{def_bn}) as always. \Cref{fig_A44} depicts the construction. The numerator factors $(x_1+x_2)(x_3+x_4)$ indicate that at least two of the four external edges need to be cancelled in order for an uncancelled internal edge to be possible.  Setting all $x_j=0$ except one and then symmetrizing  $j$ yields \cref{free_A1n} as indicated below \cref{def_Ajn},
	\begin{align*} 
	A^1_4 &= -i  b_3\left(x_1+x_2+x_3+x_4\right) .
	\end{align*}
	Note especially that $A^4_4 \equiv A^3_4 \equiv A^2_4$, i.e. there is no term of higher than second order in the external variables $x_j$. This is because of the Euler Characteristic \cref{euler_graph}, a tree with four external edges has at most two vertices and each vertex is linear in the offshell variables. Hence, the overall amplitude can at most be quadratic.  
\end{example}

\begin{figure}[h]
	\begin{tikzpicture}
	
	\node [] at (0,0) {$A^4_4\sim $};
	
	\node at (1,-.2) {$\sum\limits_{\text{4 Perm.}}$};
	\node[diffTree] (a) at (2.3,0){};
	\draw [thick] (a) to +(90:.5);
	\draw [-|, thick, bend angle=30, bend left] (a.240) to +(210:.5);
	\draw [-|, thick] (a.270) -- +(270:.5); 
	\draw [-|, thick, bend angle=30, bend right] (a.300) to +(330:.5);

	\node at (4,-.2) {$+\sum\limits_{\text{6 Perm.}}$};
	\node[diffTree] (a) at (5.3,.2){};
	\node[diffTree] (b) at (5.9,-.4){};
	\draw[thick, -|-, bend angle = 40, bend right] (a.290) to (b.160);
	\draw[thick, -|, bend angle = 20, bend left] (a.260) to +(220:.5);
	\draw[thick, -|, bend angle = 20, bend right] (b.190) to +(230:.5);
	\draw[thick] (b.0) to +(0:.4);
	\draw[thick] (a.90) to +(90:.4);

	\end{tikzpicture}
	\caption{Construction of $A^4_4$ by connecting tree sums $b_j$. The external propagators included in $b_j$ have not been taken into account graphically, hence   $\sim$ and not $=$. ``Perm.'' indicates permutations of external edges. The first sum runs over the four possibilities for one of the edges being the offshell edge of $b_3$, the second sum over all six possibilities to choose two out of four edges as the offshell ones.  }
	\label{fig_A44}
\end{figure}

\subsection{Tadpole graphs}

If $\Gamma$ is a Feynman graph, we denote the set of its internal edges by $E_\Gamma$ and the set of vertices by $V_\Gamma$. A vertex where one of the external edges of a diagram is attached is called external vertex.

\begin{definition}\label{def_tadpole}
	A tadpole graph is a Feynman graph where there is at least one closed path of edges which is connected to the rest of the diagram via at most one vertex. 
\end{definition}
In terms of Feynman integrals, this amounts to an integral over the corresponding loop momenta which is independent of any external momenta of the amplitude. Especially, we call a diagram a tadpole if there is at least one such loop, not only if it consists of a momentum-independent loop exclusively. A non-tadpole graph has at least two external vertices.
We assume that tadpole diagrams give no contribution to the $S$-matrix. This is the case in kinematic renormalization schemes  \cite{brown_angles_2011}.

By cancellation of internal edges, the diffeomorphism Feynman rules \cref{diff_vn} can turn a non-tadpole diagram into a tadpole by two mechanisms:
\begin{enumerate}
	\item They cancel all but one edges in any loop in the diagram. By this, the loop only contains a single vertex and is a tadpole.
	\item They cancel a path of edges which connects all external vertices. By cancellation, the path collapses to a single effective vertex where all external momenta are attached, that is, all integrals of the amplitude will be independent of external momenta. 
\end{enumerate}
Both effects are illustrated in \cref{fig_tadpoles} for a 2-loop example graph.

\begin{figure}[h]
	
	\centering
	\begin{tikzpicture}

	\node [diffVertex] (a) at (1.5,3) [] {};
	\node [diffVertex] (b) at (3,3.8) [] {};
	\node [diffVertex] (c) at (3,2.2) [] {};
	\node [diffVertex] (d) at (4.5,3) [] {};
	
	\draw [thick,-|-] (a) ..controls +(.4,.8) and +(-.4,0) ..  (b);
	\draw [thick,red] (a) ..controls +(.4,-.8) and +(-.4,0) .. (c);
	\draw [thick,-|-] (c)  -- (b);
	\draw [thick,-|-] (b) ..controls +(.4,0) and +(-.4,.8) ..  (d);
	\draw [thick,red] (c) ..controls +(.4,0) and +(-.4,-.8) ..  (d);
	
	\draw [thick,-] (a) --   +(180:.5);
	\draw [thick,-] (d) --  + (0:.5);
	\draw [thick,-] (c) --  + (270:.5);

	\node at (6,3){$\longrightarrow$};

	\node [diffVertex] (b) at (8.5,3.8) [] {};
	\node [diffVertex] (c) at (8.5,2.2) [] {};

	\draw [thick,-|-] (b) ..controls +(.7,-.2) and +(.7,.5) ..  (c);
	\draw [thick,-|-] (c)  -- (b);
	\draw [thick,-|-] (b) ..controls +(-.7,-.2) and +(-.7,.5) ..  (c);

	\draw [thick,-] (c) --   +(180:.5);
	\draw [thick,-] (c) --  + (0:.5);
	\draw [thick,-] (c) --  + (270:.5);

	\node [diffVertex] (a) at (1.5,6) [] {};
	\node [diffVertex] (b) at (3,6.8) [] {};
	\node [diffVertex] (c) at (3,5.2) [] {};
	\node [diffVertex] (d) at (4.5,6) [] {};
	
	\draw [thick,red] (a) ..controls +(.4,.8) and +(-.4,0) ..  (b);
	\draw [thick,-|-] (a) ..controls +(.4,-.8) and +(-.4,0) ..  (c);
	\draw [thick,red] (c)  -- node[left]{ }(b);
	\draw [thick,-|-] (b) ..controls +(.4,0) and +(-.4,.8) ..  (d);
	\draw [thick,-|-] (c) ..controls +(.4,0) and +(-.4,-.8) ..  (d);
	
	\draw [thick,-] (a) -- node[below]{ } +(180:.8);
	\draw [thick,-] (d) -- node[below]{ } + (0:.8);
	\draw [thick,-] (c) --  + (270:.5);
	
	\node at (6,6){$\longrightarrow$};
	
	\node [diffVertex] (a) at (8,6) [] {};
	
	\node [diffVertex] (d) at (10,6) [] {};

	\draw [thick,-|-] (a) ..controls +(.8,1.2) and +(-1,1.2) ..  (a);
	\draw [thick,-|-] (a) ..controls +(.4,.5) and +(-.4,.5) ..  (d);
	\draw [thick,-|-] (a) ..controls +(.4,-.5) and +(-.4,-.5) ..  (d);
	
	\draw [thick,-] (a) --  +(180:.8);
	\draw [thick,-] (d) --  + (0:.8);
	\draw [thick,-] (a) --  + (270:.5);

	\end{tikzpicture}
	\caption{Two different ways to obtain tadpoles from non-tadpole graphs by cancellation of internal edges. Cancelled edges are red, uncancelled ones have a perpendicular line. First row: cancellation of all but one edge in a loop, second row: Cancellation of a path connecting all external vertices.}
	
	\label{fig_tadpoles}
\end{figure}
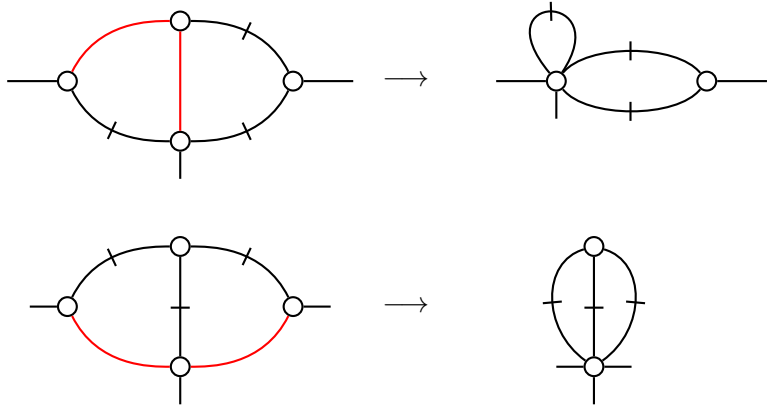

As a result of the above discussion, we have the following lemma:
\begin{lemma}\label{lem_tadpole}
	Assume $\Gamma$ is no tadpole graph, then
	\begin{enumerate}
		\item There are at least two uncancelled edges  in any closed path of edges in $\Gamma$.
		\item There is at least one external  vertex $v$ such that there is not a path of cancelled edges connecting $v$ to all other external vertices.
	\end{enumerate}
\end{lemma}

\FloatBarrier

\subsection{Loop amplitudes}\label{subs_free_loops}
We will see in the following that the tree sums $A^j_n$ from \cref{def_Ajn} act as building blocks of Feynman diagrams by the mechanism discussed in \cref{subs_cuts}. Since they take the role of $n$-valent vertices in these diagrams, we call them ``meta-vertices''. Note that this does not mean that  $A^j_n$ collapses to an actual vertex inside the Feynman diagram, it might still contain uncancelled internal edges if $j>1$.

\begin{lemma}\label{lem_loop}
	The Feynman integrand of the sum $G$ of all connected graphs with $l>0$ loops is obtained by building all $l$-loop graphs from  $n_k$-valent meta-vertices with Feynman amplitude $A^{j_k}_{n_k}$ from \cref{def_Ajn}. Here, $j_k$ is the number of external edges of $G$ connected to the meta-vertex $A^{j_k}_{n_k}$. Each pair of meta-vertices is connected by at least two internal edges.
\end{lemma}
\begin{proof}
	The proof consists of three steps: First, we show that any non-tadpole graph can be decomposed into trees. Second, we specify how the individual components turn into tree sums $A^j_k$ acting as meta-vertices when summing over all graphs. Third, we show that no pair of meta-vertices must be connected by only a single internal edge. 
	
(1): Let $\Gamma$ be any non-tadpole Feynman graph according to \cref{def_tadpole}. Then by \cref{lem_tadpole} there exists at least one set of uncancelled edges $U\subseteq E_\Gamma$. Let $U$ be such a set with minimum number of edges. Identify the uncancelled edges as cuts following \cref{subs_cuts}. These cuts divide $\Gamma$ into connected components $\Gamma_1, \ldots, \Gamma_n$. By \cref{lem_tadpole}, $n\geq 2$ and no $\Gamma_k$ contains loops, i.e. $\Gamma_k$ are trees. Let $n_k$ be the number of external edges of $\Gamma_k$. Let further $j_k$ be the number of external offshell edges of $\Gamma_k$, i.e. which were external edges of the uncut graph $\Gamma$. 

(2): Now sum  all possible ways $U$ of cutting $\Gamma$ with the minmal possible number of cuts and for each cut, replace all  $\Gamma_k$   by the corresponding tree sum $A^{j_k}_{n_k}$ from \cref{def_Ajn}. By \cite[Prop.~4.1]{KY17}, this sum equals the sum of all connected graphs with $\abs{\Gamma}$ loops, weighted with their correct symmetry factors up to an overall factor. In other words: The total non-tadpole integrand of an amplitude with $\abs{\Gamma}$ loops is the sum of all possible ways of connecting the $(n_k-j_k)$ onshell edges of tree sums $A^{j_k}_{n_k}$. Thereby, these tree sums act as meta-vertices.  No two edges of the same meta vertex $A^{j_k}_{n_k}$ must be connected to each other and the $j_k$ offshell edges become external edges of $\Gamma$.  In case $j_k=0$ for some $k$, that meta vertex vanishes by \cref{lem_Sn_free}. Note that regardless of \cref{def_Ajn}, if $j_k>0$ then at least one of those edges actually \emph{is} offshell because if all are onshell the amplitude is zero by \cref{lem_Sn_free}. Hence each non-vanishing contribution to $\Gamma$ has at least one external edge at each meta vertex.

(3): Assume there is a pair of meta-vertices $A^{j_k}_{n_k}$, $k\in \{1,2\}$ connected by only a single internal edge $e$.  Since $\Gamma$ is no tree, there is at least one other path between these meta-vertices, involving yet another vertex. This forms a loop with at least 3 uncancelled edges. So $e$ can be cancelled without making $\Gamma$ a tadpole, i.e. $U$ does contain more than the minimum necessary number of cuts. Indeed, $A^{j_1}_{n_1} \frac{i}{x_e}A^{j_2}_{n_2}$ together with the case where $e$ is cancelled form the tree sum $A^{j_1+j_2}_{n_1+n_2-2}$ which is taken into account when a cut $U$ without cutting $e$ is used. 
\end{proof}

Usually, if in a Feynman graph two external momenta enter at the same vertex $v$, they can be added to yield a single effective momentum. This is not the case here: A meta-vertex $A^{j>1}_n$  has internal uncancelled edges, so if two external momenta are entering $A^j_n$, they can not be combined into only a single external edge. Only if $j=1$, the meta-vertex $A^1_n$ has no internal structure by \cref{thm_bn}.

\begin{example}
	Consider $A^2_4=A^4_4$ from \cref{ex_A4}. Assume for simplicitly the two external offshell edges are numbered 1 and 2 and the remaining two edges are onshell, then this meta-vertex has an amplitude 
		\begin{align*} 
	A_2^4 \Big|_{x_3=0=x_4}      &=   -ib_3 \left( x_1+x_2  \right)     -i b_2^2\left( \frac{x_1 x_2 }{x_{1+3}} +\frac{x_1x_2}{x_{1+4}}   \right) .
	\end{align*}
	The second summand contains an uncancelled propagator and is proportional to $x_1x_2 = (p_1^2-m^2)(p_2^2-m^2)$ whereas the first summand is proportional to $p_1^2 + p_2^2-2m^2$. It is not possible to reproduce this Feynman amplitude if one ``combines'' both momenta into some effective momentum $p:= p_1+p_2$ . So if the external momenta $p_1$ and $p_2$ are incident to the meta-vertex $A^2_4$ then they have to stay distinct. 
	
\end{example}

\Cref{lem_loop} implies that an amplitude with $k$ external offshell edges contains at most $k$ meta-vertices $A^j_n$. On the other hand, by \cref{lem_tadpole} it contains at least two such vertices in order to not be a tadpole. Hence, the diffeomorphism  does not contribute to the  onshell amplitudes, i.e. the $S$-matrix \cite[Thm.~4.7]{KY17} and graphs with two external edges have the topology of $l$-loop multiedges \cite[lem.~4.8]{KY17} with vertices $A^1_n\sim b_{n-1}$.  This implies an alternative proof for \cref{lem_inverse}: 
\begin{lemma}\label{lem_inverse2}
	The coefficients $b_n$ of the inverse 
	\begin{align*}
	\rho(x) &= \sum_{n=1}^\infty \frac{b_n}{n!}\phi^n(x)
	\end{align*}
	of the diffeomorphism \cref{def_diffeomorphism} are the Feynman amplitudes of amputated meta-vertices $A^1_{n+1}$ (i.e. of the tree sums $b_n$ as defined in \cref{def_bn}).
\end{lemma}
\begin{proof}
	Consider the non-amputated time-ordered 2-point correlation function of $\rho$. By the above discussion, it is supported on  $l$-loop multiedge graphs $M^{(l)}$ where the two vertices are meta-vertices $A^1_{l+2}$ (we denote by $M^{(l)}$ the Feynman amplitude of the amputated $l$-loop multiedge). The single offshell edge of said vertices carries momentum $p$ and the graph $M^{(l)}$ has a symmetry factor  $\frac{1}{(l+1)!}$,
	\begin{align*}
	\langle \rho(-p) \rho (p) \rangle = \langle \phi(-p) \phi(p) \rangle + \frac{i}{x_p }\sum_{l=1}^\infty \left( A^1_{l+2} \right) ^2 \cdot \frac{1}{(l+1)!}M^{(l)} (p) \frac{i}{x_p}.
	\end{align*} 
	On the other hand, using the inverse diffeomorphism in position space (where we take $b_n$ to be the  unknown coefficients of the inverse) the same  function is
	\begin{align*}
	\langle T\rho(x) \rho(y)\rangle &= \sum_{j=1}^\infty \sum_{k=1}^\infty \frac{b_j b_k}{j! k!} \langle T\phi^j(x) \phi^k(y) \rangle.
	\end{align*}
	The right hand side are $(j+k)$-point correlation functions of the field $\phi$ which are computed via Wick's theorem, i.e. all factors $\phi$ of the field have to be connected in pairs which eliminates all summands were  $j+k$ is odd. If $j+k$ is even, any term where two fields at the same spacetime point are connected correspond - after Fourier transformation - to a tadpole graph which is assumed to vanish. Hence all non-vanishing pairs have to be of the form $\phi(x) \phi(y)$ which implies $j=k$,  
	\begin{align*}
	\langle T\rho(x) \rho(y)\rangle &= \langle T\phi(x) \phi(y) +  \sum_{k=2}^\infty \frac{\left( b_k\right)^2}{(k!)^2 } \langle{T\phi^k(x) \phi^k(y)} \rangle .
	\end{align*}
	There are precisely $k!$ equivalent ways of forming pairs which cancels one factor in the denominator. The resulting sum equals after Fourier transformation the sum over all $(k-1)$-loop multiedges, weighted with their correct symmetry factor $\frac{1}{k!}$,
	\begin{align*}
	\langle \rho(p) \rho(-p)\rangle &= \langle \phi(-p) \phi(p) \rangle+ \sum_{k=2}^\infty \frac{\left( b_k\right)^2}{k!} M^{(k-1)}(p).
	\end{align*}
	Comparing coefficients yields the claimed equality
	\begin{align*}
	A^1_{l+2}\big|_{x_1=x_p\neq 0}   &= -ix_p b_{l+1}.
	\end{align*}
	Note that this is not \cref{free_A1n}: In the latter, $A^1_n$ was defined using tree sums $b_n$ whereas here $b_{l+1}$ is a diffeomorphism coefficient which we showed to coincide with a tree sum.
\end{proof}

\begin{example}
	As an example for the proof, consider $c_1(p)$, the first non-tadpole contribution to the 2-point function of $\rho$. It is a 1-loop multiedge built from two 3-valent vertices $iv_3$ from \cref{diff_vn}. These vertices have only one external edge which has the offshell variable $x_p$. Their amplitude therefore is $A^1_3$ with a fixed (non-symmetric) offshell edge.  We consider non-amputated graphs so there are also the external propagators $\frac{i}{x_p}$ and the overall momentum-space amplitude is
	\begin{align*} 
	c_1(p):=\frac 12 \left( \frac{i}{x_p} \right) ^2\left( A^1_3 \big|_{x_p\neq 0} \right) ^2 \int \d^4 k \; \frac{i}{x_k} \frac{i}{x_{p-k}} = \frac 12\left( \frac{i A^1_3|_{x_p\neq 0} }{x_p} \right)    \int \d^4 k \frac{1}{x_k x_{p-k}}.
	\end{align*}
	Define the prefactor to be $A$. 
	In position space, the same amplitude by Fourier transformation is
	\begin{align*} 
	\tilde c_1(x) &= \int \d^4 p \; c_1(p) e^{ipx} = \frac 12 A \iint \d^4 p \d^4 k \; \frac{1}{x_k}\frac{1}{x_{p-k}}e^{ipx}= \frac 12 A\int \d^4 q  e^{iqx} \frac{1}{x_q}\int \d^4 k \; \frac{1}{x_k}e^{ikx}.
	\end{align*}
	This is the product of two propagators between the same two points in spacetime. That is, $\tilde c_1(x)$ corresponds to a Wick contraction of $\phi^2(y) \phi^2(y+x)$ into two pairs, $\left( \phi(y) \phi(y+x) \right) \cdot \left( \phi(y) \phi(y+x) \right)   $. 
	
	If $b_2$ is the second coefficient of the inverse diffeomorphism, i.e. $\rho = \phi + \frac 12 b_2 \phi^2 + \ldots$, then	 we expect this amplitude $\tilde c_1(x)$ to be the product of two position-space propagators with the prefactor $	2 \left( \frac{1}{2}b_2 \right) ^2$. The additional 2 arises from two possibilities to Wick-contract the fields. This prefactor has to be $A$ so we can read off that $A^1_3\big|_{x_p\neq 0} = -ix_p b_2$, that is, the diffeomorphism coefficient $b_2$ coincides with the 3-valent tree sum as claimed.
	
\end{example}

Using Lagrange inversion \cref{lagrange_inversion} to determine the coefficients of the inverse diffeomorphism, \cref{lem_inverse2} yields the explicit formula \cref{thm_bn} for the tree sums \cref{def_bn}. Note that \cref{lem_inverse2} - unlike \cref{thm_bn} - explicitly uses the fact that tadpole graphs vanish.

\section{Local diffeomorphisms of an interacting field}\label{subs_int}

\subsection{Feynman rules}
In this section, the diffeomorphism \cref{def_diffeomorphism} is applied to a $\phi^s$-type interacting field, i.e. 
\begin{align}\label{L_phis}
\mc L_\phi  &= \frac 12  \partial_\mu \phi ( x)\partial^\mu \phi (  x) -\frac 12 m^2 \phi^2( x) -\frac{\lambda_s }{s!}\phi^s( x).
\end{align}
The coupling constant $\lambda_s$ has an index $s$  to better keep track of the type of interaction.

Since the first part of \cref{L_phis} coincides with \cref{Lfree}, the field $\rho$ obtains the same vertices $i v_n$ from \cref{diff_vn} as in the free case. Additionally, the interaction monomial in \cref{L_phis} gives rise to a second type of vertex with Feynman rule
\begin{align*}
- iw^{(s)}_n &= -i\frac{\lambda_s }{s!}n! \underbrace{\sum_{j=0}^{ n-s} \sum_{k=0}^{n-s-j} \cdots \sum_{l=0}^{n-s-j-k-\ldots}}_{s-1 \text{ sums}} \underbrace{a_j a_k \cdots a_l a_{n-s-j-k-\ldots -l}}_{s \text{ factors }a} .
\end{align*}
The vertex Feynman rule is by definition $n!$ times the coefficient of $\rho^n$ of the power series $ \frac{\lambda_s}{s!} \phi^s(\rho)$. Hence using \cref{faadibruno}, Bell polynomials allow for a condensed notation of this multisum:
\begin{align}\label{diff_arb_wns}
-i w^{(s)}_n &=- i\lambda_s B_{n,s}(1! a_0, 2! a_1, 3! a_2, 4! a_3, \ldots ), \qquad w_n^{(s)}=0 \ \forall  n<s.
\end{align}
Again, as the diffeomorphism \cref{def_diffeomorphism} is tangent to identity, it reproduces the lowest order term, i.e. there is a vertex $-i w^{(s)}_s = -i\lambda_s$ and no vertex $-iw^{(s)}_n$ where $n<s$. As an illustration, the diffeomorphism vertices of underlying $\phi^3$ and $\phi^4$-theory are shown in \cref{phis_vertex_picture}.

\begin{figure}[htbp]
	\centering
	\begin{tikzpicture}

	\node at (-2.5,1)[anchor=west,text width =6cm] { underlying $\phi^3$ theory $(s=3)$};
	\node at (5,1)[anchor=west,text width =6cm]  { underlying $\phi^4$ theory $(s=4)$};

	\node at (-1.5,0) {$-iw^{(3)}_3=$};
	
	\node [intVertex]  (c) at (0,0) {} ;
	
	\draw [thick] (c) -- + (90:.5);
	\draw [thick] (c) -- + (210:.5);
	\draw [thick] (c) -- + (330:.5);
	
	\node at (.5,0) [anchor=west,text width = 4cm] {$=-i\lambda_3 $};

	\node at (-1.5,-1.5) {$-iw^{(3)}_4=$};
	
	\node [intVertex]  (c) at (0,-1.5) {} ;
	
	\draw [thick] (c) -- + (45:.5);
	\draw [thick] (c) -- + (135:.5);
	\draw [thick] (c) -- + (225:.5);
	\draw [thick] (c) -- + (315:.5);
	
	\node at (.5,-1.5)  [anchor=west,text width = 4cm]{$=  - 12 i \lambda_3  a_1 $};

	\node at (-1.5,-3) {$-iw^{(3)}_5=$};
	
	\node [intVertex]  (c) at (0,-3) {} ;
	
	\draw [thick] (c) -- + (90:.5);
	\draw [thick] (c) -- + (162:.5);
	\draw [thick] (c) -- + (234:.5);
	\draw [thick] (c) -- + (306:.5);
	\draw [thick] (c) -- + (18:.5);
	
	\node at (.5,-3) [anchor=west,text width =4cm]{$= -60 i \lambda_3 \left( a_2 + a_1^2\right)  $};

	\node at (-1.5,-4.5) {$-iw^{(3)}_6=$};
	
	\node [intVertex]  (d) at (0,-4.5) {} ;
	
	\draw [thick] (d) -- + (30:.5);
	\draw [thick] (d) -- + (90:.5);
	\draw [thick] (d) -- + (150:.5);
	\draw [thick] (d) -- + (210:.5);
	\draw [thick] (d) -- + (270:.5);
	\draw [thick] (d) -- + (330:.5);
	
	\node at (2,-4.3) {\vdots};
	
	\node at (5,0)[anchor=west,text width =5cm] {$-iw^{(4)}_3=0$};

	\node at (5,-1.5)[anchor=west,text width =5cm] {$-iw^{(4)}_4=  - i\lambda_4$};

	\node at (5,-3)[anchor=west,text width =5cm] {$-iw^{(4)}_5= -20 i \lambda_4 a_1$};

	\node at (5,-4.5)[anchor=west,text width =6cm] {$-i w^{(4)}_6=- 60 i\lambda_4 \left( 2 a_2 + 3 a_1^2\right)  $};

	\end{tikzpicture}
	\caption{Diffeomorphism-interaction vertices according to \cref{diff_arb_wn}. The graphical representation is black dots to distinguish them from "pure" diffeomorphism vertices \cref{free_vertex_picture}.}
	\label{phis_vertex_picture}
\end{figure}

Analogous to \cref{L_free_rho}, the $w^{(s)}_n$ from \cref{diff_arb_wn} can be interpreted as coupling constants, namely the Lagrangian density of $\rho$ reads 
\begin{align} \label{L_phis_rho}
\mc L_\rho &= -  \sum_{n=1}^{\infty}  \frac{f_{n+1}}{n!} \rho^{n}(x) \partial_\mu \partial^\mu \rho (x)   + m^2\sum_{n=2}^{\infty} \frac{  c_{n-2}}{n!} \rho^{n}(x) -\sum_{n=s}^\infty \frac{w^{(s)}_n }{n!}\rho^n(x).
\end{align}
 
The above construction works equivalently in the case of multiple interaction terms in the original Lagrangian density,
\begin{align}\label{L_arb}
\mc L_\phi  &= -\frac 12 \phi (  x) \partial_\mu \partial^\mu \phi ( x) -\frac 12 m^2 \phi^2(  x) -\sum_{s=3}^\infty \frac{\lambda _s}{s!}\phi^s( x).
\end{align}
Each interaction monomial gives rise to a new family of interaction vertices \cref{diff_arb_wns}. 
The total contribution to the $n$-valent interaction vertex is the sum of all $-iw_n^{(s)}$ where $s \leq n$,
\begin{align}\label{diff_arb_wn}
-i w_n &:= -i\sum_{s=3}^n w_n^{(s)} = -i \sum_{s=3}^n \lambda_s B_{n,s} \left( 1! , 2! a_1, 3! a_2, \ldots  \right) .
\end{align}

Together with the diffeomorphism vertex $iv_n$ from \cref{diff_vn}, the most general form of a  $n$-valent vertex, which can arise through local diffeomorphism of a scalar theory \cref{L_arb}, is
\begin{align}\label{arb_generalvertex}
&iv_n-i w_n    \\
&= i \Big(B_{n-2,1} (  2! a_1, 3! a_2, \ldots) + B_{n-2,2} (2! a_1, 3! a_2, \ldots )\Big)\left( x_1 + \ldots + x_n  \right) \nonumber  \\
&\quad + i m^2\Big( n B_{n-2,1} (  2! a_1,  \ldots) + nB_{n-2,2} (2! a_1, \ldots )  -B_{n,2} \left( 1! a_0, 2! a_1, \ldots  \right)\Big)  \nonumber \\
&\quad  -i \sum_{s=3}^n \lambda_s B_{n,s} \left( 1! , 2! a_1, 3! a_2, \ldots  \right). \nonumber 
\end{align}
Note that - althought there are two possibly infinite sets of free parameters $\{\lambda_s\}_{s \geq 3}$ and$\{a_j \}_{j\geq 1}$ - the momentum dependence of this vertex is very restricted. This fact is a consequence of locality of the diffeomorphism: It does not introduce additional derivatives, and since the underlying Lagrangian  \cref{L_arb} only contains up to second derivatives, no local diffeomorphism can produce a vertex with Feynman rule of higher than second order in momenta. Compare \cref{subs_nonlocal}.

\subsection{Cancellation of higher interaction vertices}\label{subs_cancellation_higher}

As soon as interaction is present in the underlying Lagrangian, by \cref{diff_arb_wn} an infinite set of new interaction vertices is generated by the diffeomorphism. Consider the sum of all tree graphs with $n$ external edges  instead of only the vertex $-i w_n^{(s)}$, then almost all contributions cancel if the external edges are onshell. The only remaining terms arise from trees build of $s$-valent vertices $-i \lambda_s$. For underlying $\phi^4$-theory, this statement is  \cite[thm. 5.1]{KY17}. Note also in \cite{velenich} this cancellation is discussed based on properties of the $S$-matrix, here we instead use explicit identities for the Feynman rules.

\begin{example} Consider the tree sum with $n=4$ external edges in the diffeomorphism of $\phi^3$ theory, $s=3$. We restrict to terms proportional to $\lambda_3$, these are trees which contain precisely one vertex of type $-iw^{(3)}_j$ (and allother vertices are of diffeomorphism type $iv_k$). There are two contributions: 
	~ \\
	
	{
		\centering
		\begin{tikzpicture}  
		
		\node at (-1.5,0) {$S^{(3)}_4=$};

		\node [intVertex]  (c) at (0,0) {} ;
		
		\draw [thick, -|] (c) -- + (45:.5);
		\draw [thick, -|] (c) -- +(135:.5);
		\draw [thick, -|] (c) -- +(225:.5); 
		\draw [thick, -|] (c) -- +(315:.5); 
		
		\node at (0,-.9){$\underbrace{\qquad \quad  }_{S^{(3)}_{4,4}}$};

		\node at (1.8,-.15) {$+ \quad \sum \limits_{\text{Perm.}} $};
		\node [intVertex] (c1) at (3.5,.3) {};
		\node [diffTree] (c2) at (3.5,-.2) {};
		\draw [thick] (c1) -- (c2);
		\draw [thick, -|] (c1) -- + (45:.5);
		\draw [thick,-|] (c1) -- +(135:.5);
		\draw [thick,bend angle=30, bend left, -|] (c2.250)  to +(225:.5); 
		\draw [thick,bend angle=30, bend right, -|] (c2.290)  to +(315:.5);

		\node at (3.5,-1.2){$\underbrace{\hspace{1cm}   }_{S^{(3)}_{4,3,\text{single}}}$};

		\end{tikzpicture}
		
	}
	$S^{(3)}_{4,4}$ is the 4-valent interaction vertex
	\begin{align*}
	S^{(3)}_{4,4} = -iw^{(3)}_4 = -12i\lambda_3 a_1.
	\end{align*}
	 The second contribution to $S_4$ is a sum over trees consisting of one interaction vertex and a tree sum $b_2$ (which happens to be the vertex $iv_3$ here). There are six different permutations of the external edges, one of them is 
	\begin{align*}
	S^{(3)}_{4,3,\text{single}} &= -i\lambda_3 \cdot \left( -2a_1 \right)   = 2i\lambda_3 a_1.
	\end{align*}
	This does not depend on external momenta, therefore summing the six permutations produces
	\begin{align*}
	S^{(3)}_{4,3} &= 6 \cdot S^{(3)}_{4,3,\text{single}} = 12i\lambda_3 a_1.
	\end{align*}
	Hence, the total amplitude proportional to $\lambda_3$ is
	\begin{align*}
	S^{(3)}_4 &= S^{(3)}_{4,4} + S^{(3)}_{4,3} =0.
	\end{align*}
\end{example}

Now let $s\geq 3$ be arbitrary but fixed. 
We will show the analogue of \cref{lem_Sn_free} for the interacting case, namely that all $S^{(s)}_n$ for $n >s$ are zero. This is based on two observations: 
\begin{enumerate}
	\item Any interaction vertex $-i w^{(s)}_n$ is of order one in $\lambda_s$. Hence, it has to be cancelled against trees which again contain a single vertex of $-i w^{(s)}_j$ type and the remaining vertices are of pure diffeomorphism type $i v_j$.  
	\item An interaction vertex $-iw^{(s)}_n$ does not cancel adjacent propagators, therefore such vertex can only be cancelled against a tree sum which also does not cancel propagators. Especially, one can require all $n$ external edges to be on-shell.  
\end{enumerate}
Euler characteristic \cref{euler_graph} ensures compatibility of these requirements: A tree $T$ with $\abs{V_T}$ vertices has $\abs{V_T}-1$ internal edges, and if one of the vertices is of $-i w^{(s)}_j$ type, the remaining $\abs{V_T}-1$ vertices of $i v_j$ type precisely suffice to cancel all internal edges but no external one. Summing over all possible trees and permutations of external edges, $S_n$ consists of summands where one vertex $-iw^{(s)}_j$ is connected to $j$ tree sums $b_{k_1}, \ldots, b_{k_j}$ such that $k_1 + \ldots + k_j=n$. This is shown in \cref{Sn_image} for $s=3$ (for other $s$, the sum would end at an $s$-valent interaction vertex on the right side).

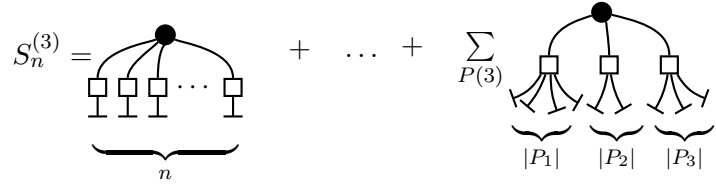
\begin{figure}[htbp]
	\centering
	\begin{tikzpicture}

	\node at (0,0) {$S^{(3)}_n=$};
	
	\node [intVertex] (c1) at (1.5,.2) {};
	\node [diffTree] (v1) at (.6, -.5){};
	\node [diffTree] (v2) at (1, -.5){};
	\node [diffTree] (v4) at (1.4, -.5){};
	\node   at (1.9, -.5){$\cdots $};
	\node [diffTree] (v3) at (2.4, -.5){};
	
	\draw [thick ] (c1) .. controls +(-.2,-.1) and +(0,.5).. (v1);
	\draw [thick ] (c1)  .. controls +(-.2,-.2) and +(0,.3)..(v2);
	\draw [thick ] (c1)  .. controls +(.0,-.1) and +(0,.5).. (v4);
	\draw [thick ] (c1)  .. controls +(.2,-.1) and +(0,.5).. (v3);
	
	\draw [-|,thick] (v1) -- +(270:.4);
	\draw [-|,thick] (v2) -- +(270:.4); 
	\draw [-|,thick] (v3) -- +(270:.4);
	\draw [-|,thick] (v4) -- +(270:.4);
	
	\node at (1.5, -1.5) {$\underbrace{\qquad  \qquad  \quad }_{n}$};
	
	\node at (3,0) [anchor=west,text width = 1cm] {$+\quad \ldots  $};

	\node at (4.5,-.15) [anchor=west,text width = 1cm] {$+ \quad \sum \limits_{P(3)} $};
	
	\node [intVertex] (c1) at (7.3,.5) {};
	\node [diffTree] (v1) at (6.6, -.2){};
	\node [diffTree] (v2) at (7.4, -.2){};
	\node [diffTree] (v3) at (8.2, -.2){};

	\draw [thick ] (c1) .. controls +(-.3,-.1) and +(0,.5).. (v1);
	\draw [thick ] (c1)  .. controls +(.05, -.1) and +(0,.3)..(v2);
	\draw [thick ] (c1)  .. controls +(.4,-.1) and +(0,.5)..(v3);
	
	\draw [-|, thick, bend angle=20, bend left] (v1.230) to +(220:.5);
	\draw [-|, thick, bend angle=10, bend left] (v1.250) to +(240:.5); 
	\draw [-|, thick, bend angle=10, bend right] (v1.280) to +(280:.5); 
	\draw [-|, thick, bend angle=20, bend right] (v1.300) to +(310:.5); 
	
	\node at (6.5, -1.3) {$\underbrace{\qquad    }_{\abs{P_1}}$};
	
	\draw [-|, thick, bend angle=10, bend left] (v2.260) to +(250:.5);
	\draw [-|, thick, bend angle=10, bend right] (v2.280) to +(290:.5); 
	
	\node at (7.5, -1.3) {$\underbrace{\quad }_{\abs{P_2}}$};

	\draw [-|, thick, bend angle=10, bend left] (v3.260) to +(250:.5);
	\draw [-|, thick, bend angle=10, bend right] (v3.280) to +(290:.5); 
	\draw [-|, thick, bend angle=20, bend right] (v3.310) to +(320:.5); 
	
	\node at (8.4, -1.3) {$\underbrace{\qquad  }_{\abs{P_3}}$};

	\end{tikzpicture}
	\caption{Structure of the contributions to $S^{(3)}_n$: Since the external edges are onshell, all terms consist of tree sums $b_k$ and one  interaction vertex $-iw^{(3)}_k$. The sum $P(k)$ runs over all possible ways of distributing the $n$ external edges to the given number $k$ of tree sums $b_{\abs{k_i}}$. }
	\label{Sn_image}
\end{figure}

The first term only contains tree sums $b_1\equiv 1$, it is
\begin{align*}
S^{(s)}_{n,n}&=-i w^{(s)}_n \underbrace{b_1 b_1 \cdots b_1}_{n \text{ factors}}.
\end{align*}
The second contribution to $S^{(s)}_n$ contains one $b_2$ and the rest $b_1$, it is $S^{(s)}_{n,(n-1), \text{single}} = -iw^{(s)}_{n-1} b_2  \cdot b_1^{n-2}$.
The third term is made from a vertex $-iw^{(s)}_{n-2}$ and either one $b_3$ or two  $b_2$ and the rest $b_1$.
The last contribution to $S^{(s)}_n$, shown in the right in \cref{Sn_image}, has one vertex $-iw^{(s)}_s$ and $s$ tree sums $b_{k_1}, b_{k_2}, b_{k_s}$ such that $k_j \geq 1, k_1 + k_2 +\ldots +  k_s =n$. 

For any fixed valence $k$ of the interaction vertex $-i w^{(s)}_k$, there is a sum over the set $P(k)$ of all possible ways of assigning the $n$ external edges to the individual tree sums. This is the same partition as in the definition of Bell polynomials \cref{bell_partitions}, hence for any fixed $k$
\begin{align*}
S^{(s)}_{n,k} &=-i w^{(s)}_k\sum_{P(k)} S^{(s)}_{n,k,\text{single}} = -i w^{(s)}_k\sum_{P(k)} \prod_{i=1}^k b_{k_i} = B_{n,k} \left( b_1, b_2, \ldots  \right) .
\end{align*}
Finally, $S^{(s)}_n$ contains all of these terms for $k\in \{s, \ldots, n\}$. Inserting the Feynman rule \cref{diff_arb_wns} for $-iw^{(s)}_n$ yields
\begin{align}\label{diff_phis_Sn}
S^{(s)}_n &= -i \lambda_s \sum_{k=s}^n   B_{k,s}(1, 2! a_1, 3! a_2, 4! a_3, \ldots )  B_{n,k} \left( b_1, b_2, \ldots  \right).
\end{align}

\begin{theorem}\label{thm_phis_Sn}
	  $S^{(s)}_s= -i\lambda_s$ and  $S^{(s)}_n= 0$ for any $n\neq s$. 
\end{theorem}

\begin{proof}\footnote{The interpretation of $S_n^{(s)}$ as a series coefficient was originally suggested by Ali Assem Mahmoud. This simplified the proof considerably.} For $k<s$, the Bell polynomial $B_{k,s}(\ldots)$ is zero, hence the sum can be started at $k=1$. 
  $B_{k,s} (1, 2! a_1, 3! a_2, \ldots) = k! [\rho ^k] \phi^s(\rho)$ is by \cref{faadibruno} the $k^{\text{th}}$ coefficient of the power series $\phi^s(\rho)$. Using \cref{lem_inverse}, $\rho(\phi) = \sum_{j=1}^\infty \frac{b_j}{j!}\phi^j$,  the tree sum $S_n^{(s)}$ by \cref{faadibruno} can be interpreted as the coefficient of $\phi^n$  of a power series 
  \begin{align*}
  S^{(s)}_n &=-i \lambda_s \sum_{k=s}^n   k! [\rho^k] \phi^s(\rho) \cdot   B_{n,k} \left( b_1, b_2, \ldots  \right)\\
  &= -i\lambda_s[\phi^n] \left(  \big(\phi (\rho(\phi)) \big)^s \right)  =-i\lambda_s [\phi^n]\phi^s = -i\lambda_s \delta_{ns}.
  \end{align*}
  Hence the tree sum is $-i\lambda_s$ for $n=s$ and zero otherwise. 
\end{proof}

If, according to \cref{L_arb}, several different interaction terms are present in the original theory, then each of them comes with a different coupling constant $\lambda_s$ and corresponding tree sums $S^{(s)}_n$. Each tree which contains at least two different interaction vertices $w_j^{(s_1)}$ and $w_k^{(s_2)}$ can be decomposed into a product of trees containing only a single one, connected via an intermediate propagator, see \cref{subs_cuts}. This is always possible since an interaction vertex cannot cancel an adjacent propagator, so there is always (at least) one uncancelled propagator between any two interaction vertices. 
The total tree sum linear in the coupling constants consequently is a sum of the individual contributions
\begin{align}\label{diff_arb_Sn}
S_n &:= \sum_{s=3}^\infty S_n^{(s)} = -i \sum_{s=3}^n \lambda_s \sum_{k=s}^n  B_{n,k} \left( b_1, b_2, \ldots  \right)  B_{k,s}(1, 2! a_1, 3! a_2, 4! a_3, \ldots ).
\end{align}
From \cref{thm_phis_Sn} one has
\begin{align}\label{Sn_lambda} 
S_n &= -i \lambda_n \qquad \forall n \geq 3.
\end{align}

\subsection{Tree sums including interaction}\label{subs_sum_int}

In this subsection, we restrict ourselves to only  one arbitrary but fixed power $s$ of interaction monomials in the underlying Lagrangian density \cref{L_arb}. Like in \cref{subs_cancellation_higher}, the case of multiple interaction monomials is a straightforward generalization.

For the diffeomorphism of $\phi^s$-theory \cref{L_phis} there are additional vertices $-iw^{(s)}_n$ from \cref{diff_arb_wns} which contribute to tree sums.  Therefore the tree sums $b_n$ from \cref{def_bn}  need to be generalized to comprise both vertex types:
\begin{definition} \label{def_bsn}
	$b'_n$ is the sum of all trees with $n$ external onshell edges without propagator and additionally  one external offshell edge with corresponding propagator. $b'_n$ consists of   both pure diffeomorphism vertices $iv_j$ and interaction-diffeomorphism vertices $-iw^{(s)}_j$.
\end{definition}

The tree sum $b'_n$ includes the trees  where all vertices are of diffeomorphism type $iv_k$, therefore $b'_n = b_n + \mc O \left( \lambda_s \right)  $. The cancellation of higher interaction vertices \cref{thm_phis_Sn} implies that only one additional vertex type, $-iw_s^{(s)} = -i\lambda_s$, needs to be included into the trees to build $b'_n$. Namely, let there be any vertex $-i w_j^{(s)}$, $j>s$, then one of two cases can occur: 
\begin{enumerate}
	\item All (internal or external) edges connected to $-iw^{(s)}_j$ are uncancelled. Then, in the sum over all possible trees, there will be precisely all trees at the position of $-iw^{(s)}_j$ to make up $S_j$,compare \cref{Sn_image}, and this sum vanishes due to \cref{thm_phis_Sn}. Hence all vertices $-iw^{(s)}_{j>s}$ where no adjacent edge is cancelled can be left out from the beginning. 
	\item At least one adjacent edge to $-iw^{(s)}_j$ is cancelled. This cancellation always originates from an adjacent diffeomorphism vertex, not from $-iw^{(s)}_j$ itself. But then, $-iw^{(s)}_j$ and the adjacent vertex  form a contribution to $S^{(s)}_m$ for some $m>j>s$. Since $S^{(s)}_m=0$, even these contributions can be left out.
\end{enumerate}
All that remains  are vertices $-iw^{(s)}_s=-i\lambda_s$ - which are the original vertices of $\phi^s$-theory - where none of the $s$ adjacent edges are cancelled.

By \cref{def_bsn}, there is only one offshell external edge in $b'_n$, the top one. But by \cref{lem_Sn_free}, any tree of pure diffeomorphism type vanishes if all its external edges are onshell. Hence any contribution of pure diffeomorphism type must be connected to the single offshell edge of $b'_n$. Consequently, all vertices of pure diffeomorphism type $iv_j$ are connected inside $b'_n$. On the other hand, the vertices $-i\lambda_s$ do not need to be connected to each other, they can be attached anywhere to the pure diffeomorphism part. This determines the structure of $b'_n$ up to   a summation over all permutations of external edges: 
\begin{lemma}\label{lem_bsn}
	The $b'_n$ are of the following form: 
	\begin{align*} 
	b'_k &= b_k \qquad \forall\; k<s-1\\
	b'_{s-1} &= b_{s-1} + \frac{\lambda_s}{x_{1+\ldots + (s-1)}}.
	\end{align*}
	If $n \geq s-1$ then there is a tree sum $b'_{s-1}$ where the single offshell edge is the offshell edge of $b'_n$ and suitably many interaction vertices $-i\lambda_s$ are attached via uncancelled internal propagators to the lower edges of $b'_{s-1}$ to make up $n$ external edges. Further, a sum over all possible $b_k$ with $k\neq s-1$ where again the appropriate number of interaction vertices are attached to the lower edges.  
\end{lemma}

\begin{example} \label{ex_bsn}
	Consider   a diffeomorphism of $\phi^3$-theory, i.e. $s=3$. An explicit calculation yields
	\begin{align*} 
	b'_2(1,2) &= b_2 +(-i\lambda_3)\frac{i }{x_{1+2}} = -2a_1 + \frac{\lambda_3}{x_{1+2}} =b'_2(1+2)\\
	b'_3 (1,2,3)&=b_3 +\left( -2 a_1 + \frac{\lambda_3}{x_{1+2+3}} \right)     (-i\lambda_3) \left( \frac{i}{x_{2+3}} + \frac{i}{x_{1+3}} + \frac{i}{x_{1+2}} \right)  \\
	&=b_3 +b'_2(1+2+3) \cdot (-i\lambda_3) \left( \frac{i}{x_{2+3}} + \frac{i}{x_{1+3}} + \frac{i}{x_{1+2}} \right)  
	\end{align*}
	\begin{align*}
	b'_4(1,2,3,4) &=  b_4  + b_3  \cdot (-i\lambda_3)\left( \frac{i}{x_{1+2}} + \text{ 5 more} \right) \\
	&\qquad +b'_2 (1+2+3+4)\cdot (-i\lambda_3)^2  \left(\frac{i^2}{x_{1+2}x_{3+4}} + \text{ 2 more} \right)\\
	&\qquad + b'_2(1+2+3+4) \cdot (-i\lambda_3)^2 \left( \frac{i^2}{x_{1+2}x_{1+2+3}}+ \text{ 5 more} \right) .
	\end{align*}
  The construction of $b'_n$ according to \cref{lem_bsn} is shown in \cref{fig_bsn}.	
  
 	They can be interpreted in terms of glued tree sums. The amplitude of $b'_3$ contains (besides $b_3$) a sum of products $b'_2 \cdot (-i\lambda_3)\cdot \frac{i}{x_e}$. These represent graphs where an interaction vertex $-i\lambda_3$ is attached to a tree sum $b'_2$ as shown in \cref{fig_bsn}. The Feynman amplitude is a product, hence this graph represents a cut, the same amplitude would arise from a disjoint union of $b'_2$ and $-i\lambda_3$ and the propagator. Note especially that the internal offshell variable $x_e$ is not zero, nevertheless, $b'_2$ is the same object as if it were computed with both lower edges offshell, just $x_{1+2}$ is replaced by $x_{1+2+3}$. This is again the phenomenon discussed in \cref{subs_cuts}: In terms of Feynman rules of the adjacent vertices, an uncancelled edge is an onshell edge. 
\end{example}

\begin{figure}[h]
	\begin{tikzpicture} 
	
	\node [] at (0,0) {$b'_2=$};
	\node[intTree] (x) at (1, 0){};
	\draw[thick] (x) --   +(90:.5);
	\draw [ thick, -|,bend angle=20,bend left] (x.245) to +(240:.4);
	\draw [thick, -|, bend angle=20, bend right] (x.295) to +(300:.4); 
	
	\node at (2,0) {$=$};
	\node[diffTree] (x) at (3,0)  {};
	\draw[thick] (x) --   +(90:.5);
	\draw [ thick, -|,bend angle=20,bend left] (x.245) to +(240:.4);
	\draw [thick, -|, bend angle=20, bend right] (x.295) to +(300:.4);

	\node[anchor=west] at (3.7,0) {$+$};
	
	\node[intVertex] (b) at (5,0){};
	\draw[thick] (b) --   +(90:.5);
	\draw [thick, -|] (b) -- +(300: .5);
	\draw [thick, -|] (b) -- +(240: .5);

	\node [] at (0,-2) {$b'_3=$};
	\node[intTree] (x) at (1, -2){};
	\draw[thick] (x) --   +(90:.5);
	\draw [-|, thick, bend angle=20, bend left] (x.240) to +(230:.5);
	\draw [-|, thick] (x.270) -- +(270:.5); 
	\draw [-|, thick, bend angle=20, bend right] (x.300) to +(310:.5); 
	
	\node at (2,-2) {$=$};
	\node[diffTree] (x) at (3,-2)  {};
	\draw[thick] (x) --   +(90:.5);
	\draw [-|, thick, bend angle=20, bend left] (x.240) to +(230:.5);
	\draw [-|, thick] (x.270) -- +(270:.5); 
	\draw [-|, thick, bend angle=20, bend right] (x.300) to +(310:.5); 
	
	\node[anchor=west] at (3.7,-2) {$+ \sum \limits_{\text{3 Perm.}}$};
	
	\node[intTree] (a) at (5.9,-1.7) {};
	\node[intVertex] (b) at (5.6,-2.3){};
	\draw[thick] (a) --   +(90:.5);
	\draw[thick, -|-, bend angle =10, bend left] (a.240) to (b);
	
	\draw [-|, thick, bend angle=10, bend right] (a.300) to +(300:.8); 
	\draw [thick, -|] (b) -- +(300: .5);
	\draw [thick, -|] (b) -- +(240: .5);

	\node [] at (0,-4) {$b'_4=$};
	\node[intTree] (x) at (1, -4){};
	\draw[thick] (x) --   +(90:.5);
	\draw [-|, thick, bend angle=20, bend left] (x.240) to +(220:.5);
	\draw [-|, thick, bend angle=10, bend left] (x.260) to +(250:.5);
	\draw [-|, thick, bend angle=10, bend right] (x.280) to +(290:.5);
	\draw [-|, thick, bend angle=20, bend right] (x.300) to +(320:.5); 
	
	\node at (2,-4) {$=$};
	\node[diffTree] (x) at (3,-4)  {};
	\draw[thick] (x) --   +(90:.5);
	\draw [-|, thick, bend angle=20, bend left] (x.240) to +(220:.5);
	\draw [-|, thick, bend angle=10, bend left] (x.260) to +(250:.5);
	\draw [-|, thick, bend angle=10, bend right] (x.280) to +(290:.5);
	\draw [-|, thick, bend angle=20, bend right] (x.300) to +(320:.5); 
	
	\node[anchor=west] at (3.5,-4) {$+ \sum \limits_{\text{6 Perm.}}$};
	
	\node[diffTree] (a) at (5.8,-3.7) {};
	\node[intVertex] (b) at (5.5,-4.3){};
	\draw[thick] (a) --   +(90:.5);
	\draw[thick, -|-, bend angle =10, bend left] (a.240) to (b);
	\draw[thick, -|, bend angle =20, bend right ] (a.300) to +(310: 1);
	\draw[thick, -|, bend angle =10, bend right ] (a.270) to +(290: .9);
	\draw [thick, -|] (b) -- +(300: .5);
	\draw [thick, -|] (b) -- +(230: .5);
	
	\node[anchor=west] at (6.4,-4) {$+ \sum \limits_{\text{3 Perm.}}$};
	
	\node[intTree] (a) at (8.6,-3.7) {};
	\node[intVertex] (b) at (8.2,-4.3){};
	\node[intVertex](c) at (9.0, -4.3){};
	\draw[thick] (a) --   +(90:.5);
	\draw[thick, -|-, bend angle =10, bend left] (a.240) to (b);
	\draw[thick, -|-, bend angle =10, bend right ] (a.300) to (c);
	\draw [thick, -|] (b) -- +(300: .5);
	\draw [thick, -|] (b) -- +(230: .5);
	\draw [thick, -|] (c) -- +(310: .5);
	\draw [thick, -|] (c) -- +(240: .5);
	
	\node[anchor=west] at (9.2,-4) {$+ \sum \limits_{\text{6 Perm.}}$};
	
	\node[intTree] (a) at (11.7,-3.4) {};
	\node[intVertex] (b) at (11.3,-4){};
	\node[intVertex](c) at (11.7, -4.3){};
	\draw[thick] (a) --   +(90:.5);
	\draw[thick, -|-, bend angle =10, bend left] (a.240) to (b);
	\draw[thick, -|-] (b.320) to (c);
	\draw [thick, -|, bend angle =10, bend right] (a) to +(300: 1.2);
	\draw [thick, -|] (b) -- +(240: .6);
	\draw [thick, -|] (c) -- +(310: .5);
	\draw [thick, -|] (c) -- +(240: .5);

	\end{tikzpicture}
	
	\caption{Construction of $b'_n$ for $s=3$ from \cref{ex_bsn} as  sum of $b_k$ where $k<n$ and $b'_2$ according to \cref{lem_bsn}. $\sum_\text{ n Perm.}$ indicates the sum over $n$ permutations of external edges like in \cref{free_b1_b2_b3_picture}. Crossed internal edges are uncancelled.}
	\label{fig_bsn}
\end{figure}
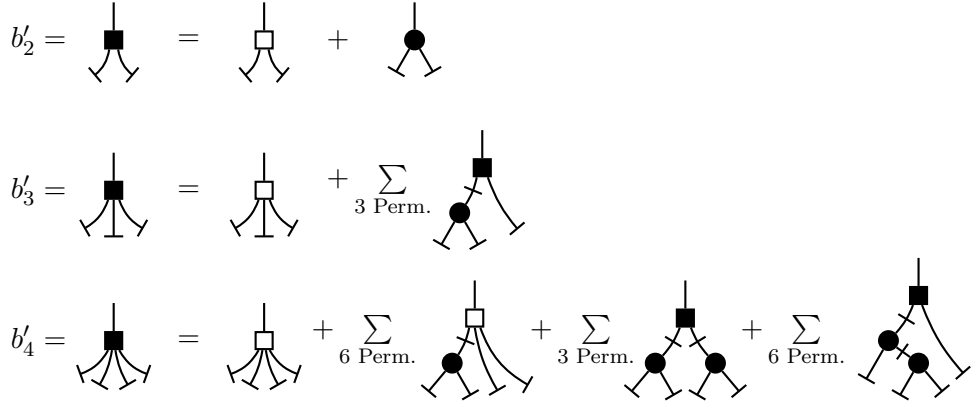

\Cref{lem_bsn}  implies that, unlike $b_n$ from \cref{thm_bn}, the $b'_n$ do depend on external momenta and masses via the remaining uncancelled internal propagators. Consequently, setting all external edges onshell does not render them zero like in \cref{lem_Sn_free}, but produces the sum of all trees which can be build from $-i\lambda_s$ vertices.

\begin{lemma}
	Let ${A'}^0_n$ be the sum of all trees of the diffeomorphism of an interacting theory with $n$ external edges, all of which are onshell. Then ${A'}^0_n$ coincides with the corresponding sum of the underlying interacting theory. That is, the diffeomorphism does not alter the onshell tree level Feynman amplitudes i.e. the $S$-matrix.
\end{lemma}
\begin{proof}
	To produce ${A'}^0_n$ whe have to use $b'_{n-1}$ but amputate the remaining offshell edge and set it onshell. Define $x_n := x_{1+\ldots + (n-1)}$ then ${A'}^0_n$ is the symmetric sum of all $n$ ways to assign the offshell edge, i.e. 
	\begin{align*} 
	{A'}^0_n &= -i x_{n} b'_{n-1} (1, 2, \ldots, n-1)\big|_{x_n=0} + \text{ (n-1)  terms where $x_n$ is exchanged for an other $x_j$}.
	\end{align*}
	 If $n-1<s$ this renders ${A'}^0_n$  zero by \cref{lem_Sn_free} since in that case $b'_{n-1} = b_{n-1}$ by \cref{lem_bsn}. 
	 
	 We concentrate on the first summand of ${A'}^0_n$ in the following, all other ones differ by just a labeling.
	 Assume $n-1= s$. 	By amputating the onshell edge any summand in $b'_{n-1}$ obtains a prefactor $-ix_{n}$. This prefactor only cancels in the summand $ \frac{\lambda_s}{x_{n}}b_{s-1}$ in $b'_{s-1}$. Hence, when $x_{n}=0$ then only $-i\lambda_s$ remains of $b'_{s-1}$ and all other contributions vanish. 
	
	Finally assume $n-1 >s$. Then again $b'_{s-1}$ collapses to $(-i\lambda_s)$. All contributions to $b'_{n-1}$ are by \cref{lem_bsn} either proportional to $b'_{s-1}$ or to some $b_k$ where $k\geq s$. All the latter vanish if the remaining edge is set onshell. What remains are the trees which are built from $b'_{n-1}=(-i\lambda_s)$ and vertices $-i\lambda_s$. This is the tree sum of the underlying $\phi^s$ theory.
\end{proof}

\subsection{Offshell tree Amplitudes}\label{sec_offshelltree}
From the tree sums $b'_n$ with one external edge offshell, one can construct tree sums with arbitrary many external edges offshell following the mechanism discussed in \cref{subs_offshell}. The  effect is similar to the inclusion of interaction vertices $-iw_n$ which do not cancel adjacent edges in \cref{subs_sum_int}:  Allowing for more external edges to be cancelled amounts to leaving more internal edges uncancelled and gives rise to tree sums which can be cut along an internal uncancelled edge according to \cref{subs_cuts}. Therefore, the structure of these tree sums is similar to \cref{lem_bsn}. 

\begin{definition}\label{def_Asjn}
	${A'}^j_n$ is defined as the sum of all trees with $n$ external edges, $j$ of which are offshell,  in the diffeomorphism of the interacting theory \cref{L_phis}. This is the generalization of \cref{def_Ajn}.
\end{definition}
Clearly, ${A'}^j_n$  contains the trees of pure diffeomorphism as a subset, so ${A'}^j_n= A^j_n + \mc O (\lambda_s)$. The interaction vertices $-i \lambda_s$ can be included in several equivalent ways. Conceptually, it is most transparent if only the tree sums $b_{s-1}$ are replaced by $b'_{s-1}$ according to \cref{lem_bsn}. This will precisely include all trees with interaction vertices but not more. Recall that by \cref{thm_phis_Sn}, it is sufficient to include the $s$-valent interaction vertices. Equivalently, one could also replace $b_{n-1}$ by $b'_{n-1}$ and leave all other $b_{k<n-1}$ intact. By \cref{lem_bsn}, this again includes all possible interaction trees. But it is not possible to do both at the same time, otherwise, the interaction trees are overcounted.

\begin{example}\label{ex_As4}

	The equivalent amplitude of \cref{ex_A4} for a diffeomorphism of $\phi^3$-theory reads
	\begin{align*}
	{A'}^4_4	&= -i\lambda_3 \left(  -i\lambda_3 -i b_2 \left( x_1+x_2+x_3+x_4 \right)  \right) \left( \frac{i}{x_{1+2}} + \frac{i}{x_{1+3}} + \frac{i}{x_{1+4}} \right) \nonumber \\
	&\qquad +A_4.
	\end{align*}
	In this case, setting $x_1=x_2=x_3=0$ and $x_4 = x_{1+2+3}\neq 0$ and including the external propagator $\frac{i}{x_{1+2+3}}$ produces
	\begin{align*} 
	& \frac{i}{x_{1+2+3}}(-i\lambda_3) \left( -i \lambda_3 -ib_2 x_{1+2+3}  \right) \left( \frac{i}{x_{1+2}} + \frac{i}{x_{1+3}} + \frac{i}{x_{1+4}} \right) \nonumber  +b_3 \\
	&\quad = (-i \lambda_3) b'_2 (1+2+3)\left( \frac{i}{x_{1+2}} + \frac{i}{x_{1+3}} + \frac{i}{x_{1+4}} \right) \nonumber  +b_3 .
	\end{align*}
	This is $b'_3(1,2,3)$ according to \cref{ex_bsn} as expected: The tree sum with one external edge offshell is obtained by taking the general tree sum and setting all but one external edge onshell. Taking the $b_2^2$ term out of $A^4_4$ and merging it with the first contribution in ${A'}^4_4$ allows to write 
	 \begin{align*} 
	 {A'}^4_4 &= \Big(-i b'_2(1+2)\Big) \Big(-ib'_2(3+4)\Big)  \frac{i}{x_{1+2}}+ \text{ 2 more permutations  }\\
	 &\qquad -i b_3\cdot (x_1+x_2+x_3+x_4). 
	 \end{align*}
	This is shown in \cref{fig_offshell}. Note especially that all contributions from interaction vertices stem from $b'_2$. It is also possible to have all interaction in the hightest possible tree sum, in this case ${A'}^4_4 \sim b'_3 + b_2\cdot b_2$. But not both $b_3$ and $b_2$ can be replaced by their interaction version simultaneously. 
	
	For 5 external edges one has
	\begin{align*} 
	{A'}^5_5 \sim b_4 + b_3 \cdot b'_2 + b'_2 \cdot b'_2 \cdot  b'_2.
	\end{align*}
	All interaction vertices could equivalently  be collected into $b'_4$, then the same amplitude reads
	\begin{align*} 
	{A'}^5_5 &\sim b'_4 + b_3 \cdot b_2 + b_2 \cdot b_2 \cdot b_2.
	\end{align*}
\end{example}

\begin{figure} [h]
	\begin{tikzpicture} 
	
	\node [] at (-.6,-2) {${A'}^4_4  $};
	\node at (0.2, -2) {$\sim$};
	
	\node at (1,-2.1) {$\sum\limits_{\text{Perm.}}$};
	\node[diffTree] (a) at (2.3,-2){};
	\draw [thick] (a) to +(90:.5);
	\draw [-|, thick, bend angle=30, bend left] (a.240) to +(210:.5);
	\draw [-|, thick] (a.270) -- +(270:.5); 
	\draw [-|, thick, bend angle=30, bend right] (a.300) to +(330:.5); 
	
	\node at (4,-2.1) {$+\sum\limits_{\text{Perm.}}$};
	\node[intTree] (a) at (5.3,-1.8){};
	\node[intTree] (b) at (5.9,-2.4){};
	\draw[thick, -|-, bend angle = 40, bend right] (a.290) to (b.160);
	\draw[thick, -|, bend angle = 20, bend left] (a.260) to +(220:.5);
	\draw[thick, -|, bend angle = 20, bend right] (b.190) to +(230:.5);
	\draw[thick] (b.0) to +(0:.4);
	\draw[thick] (a.90) to +(90:.4);
	
	\node at (0.2,-3.5) {=};
	
	\node at (1,-3.6) {$\sum\limits_{\text{Perm.}}$};
	\node[intTree] (a) at (2.3,-3.5){};
	\draw [thick] (a) to +(90:.5);
	\draw [-|, thick, bend angle=30, bend left] (a.240) to +(210:.5);
	\draw [-|, thick] (a.270) -- +(270:.5); 
	\draw [-|, thick, bend angle=30, bend right] (a.300) to +(330:.5); 
	
	\node at (3.8,-3.6) {$+\sum\limits_{\text{Perm.}}$};
	\node[diffTree] (a) at (5.1,-3.3){};
	\node[diffTree] (b) at (5.7,-3.9){};
	\draw[thick, -|-, bend angle = 40, bend right] (a.290) to (b.160);
	\draw[thick, -|, bend angle = 20, bend left] (a.260) to +(220:.5);
	\draw[thick, -|, bend angle = 20, bend right] (b.190) to +(230:.5);
	\draw[thick] (b.0) to +(0:.4);
	\draw[thick] (a.90) to +(90:.4);
	
	\end{tikzpicture}
	\caption{Construction of ${A'}^4_4$ according to \cref{ex_As4}. The interaction can either be included into the largest tree sum (here $b'_3$) or into $b'_{s-1}$ (here $b'_2$).} 
	\label{fig_offshell}
\end{figure}

\subsection{Loop amplitudes}\label{subs_interacting_loops}
The behaviour of loop amplitudes for the interacting diffeomorphism is  similar to the \cref{subs_free_loops}.

In \cref{subs_cuts} we argued that in tree graphs uncancelled internal edges amount to cutting the graph at that edge. 
If $\Gamma$ contains loops, an uncancelled internal edge can still be interpreted as a cut, but on the level of integrands. Assume $e$ is an uncancelled edge between vertices $v_1$ and $v_2$ and assume that there is another path between $v_1$ and $v_2$, not involving $e$ (i.e. $e$ is part of a loop). If $e$ is uncancelled, the integrand does not involve factors $x_e$. But then the integrand is the same one as if $v_1$ and $v_2$ would both contain one external onshell edge with momentum $p_e$ incoming or outgoing, respectively. Graphically, $e$ is cut and the two half edges become external uncancelled edges. As discussed in \cref{subs_cuts}, this does not require $x_e=0$ in the overall graph, i.e. the momentum does not actually need to be onshell. The loop integral can still be performed, assigning any value to the momentum.

\begin{theorem}\label{thm_Alm}
	Let $A^{(l)}_m$ be the sum of all connected (not only 1PI) $l$-loop Feynman integrals of the diffeomorphism of $\phi^s$-theory with $m$ external edges which are all onshell. Assume tadpole graphs vanish. Then $A^{(l)}_m$ coincides with the corresponding amplitude of $\phi^s$-theory, i.e. the diffeomorphism does not alter the $S$-matrix.
\end{theorem}
\begin{proof}
The case of trees, $l=0$, has been shown in \cref{thm_phis_Sn}.
The diffeomorphism of a free theory does not alter onshell loop amplitudes \cite[Thm.~4.7]{KY17}, hence, in order to possible influence onshell amplitudes, a graph must contain at least one interaction vertex. By \cref{thm_phis_Sn}, only the vertex $-i\lambda_s$ contributes to the sum of all integrands, the higher vertices $-iw^{(s)}_j$  \cref{diff_arb_wns} always add up to zero.  Now assume $\Gamma$ contains both interaction vertices $-i\lambda_s$ and diffeomorphism vertices $iv_n$ from \cref{diff_vn}. By \cref{thm_phis_Sn}, no edge adjacent to an interaction vertex must be cancelled, otherwise the graph is a contribution to a sum which eventually is zero. If a diffeomorphism vertex cancels an edge incident to another diffeomorphism vertex, by \cref{thm_bn} it results in a meta vertex $-ib_j x_e$ like \cref{free_A1n}, cancelling yet another adjacent edge $e$. This can be repeated, but is constrained by the Euler characteristic \cref{euler_graph}: Eventually, either all edges in a closed path or an external edge has to be cancelled. Cancelling all edges within any closed path results in a tadpole graph which vanishes by assumption.  But if all external edges are onshell, cancelling an external edge $e$ amounts to a factor $x_e=0$ and renders the integrand zero. 

Thus, only graphs without any diffeomorphism vertices contribute to the onshell non-tadpole $l$-loop amplitude. These are the graphs of the underlying $\phi^s$-theory. 
\end{proof}
The same proof holds if more than one interaction vertex is present in the underlying theory \cref{L_arb}. By \cref{Sn_lambda}, only the original vertices remain, resulting in the same Feynman graphs as in the underlying theory. 

\Cref{thm_Alm} is the graph-theoretic version of the result found in \cite{flume}. The key mechanism is very similar: A diffeomorphism gives rise to contributions proportional to the offshell-variables of the correlation functions under consideration. If all momenta are onshell, all diffeomorphism contributions obtain a zero prefactor.

As a consequence of \cref{thm_Alm}, one can view all quantum field theories as physically equivalent if they are related by a diffeomorphism. It is then possible to choose out of this equivalence class one particular representative with desirable properties. Nevertheless, the correlation functions for offshell momenta are altered by a diffeomorphism. This is exploited in the next section.

\section{Cancellation of quantum corrections}\label{sec_cancellation}

The goal of this section is to find a diffeomorphism of $\phi^s$ theory such that its  offshell correlation functions coincide with the ones of a free scalar quantum field theory. 
As mentioned in \cite{KY17}, this is especially interesting for the 2-point function, i.e. the propagator: Canonical quantization relies on the assumption that in infinite future and past, the states of an interacting theory asymptotically converge to free field states \cite[Sec. 1.4.2]{nakanishi_covariant_1990}. These states are then identified with particles observed in scattering experiments which allows to predict experimental outcomes through the LSZ-formalism \cite{lsz}. On the other hand, if any quantum field theory has the same 2-point-function as a free field, or is unitarily equivalent to a free field, this theory is itself a free field \cite{Pohlmeyer,Jost-Schroer,haag_quantum_1955}. See also  \cite{lutz} for an overview of the different no-go theorems regarding this phenomenon.

We construct a diffeomorphism of an interacting theory such that the offshell 2-point-function is free for a specific momentum. Since all diffeomorphism-related theories lie in a physical equivalence class (\cref{thm_Alm}), this might be useful to construct some new notion of asymptotic states which does not violate theorems from axiomatic quantum field theory.

\subsection{Tree level}\label{sub_cancellation_trees}

At least one external edge needs to be offshell for the diffeomorphism to possible have an influence, therefore we consider $b'_n$ from \cref{def_bsn} and  require
\begin{align*} 
b'_n &\overset = 0 \quad \forall n. 
\end{align*}
Using \cref{subs_sum_int}, this requirement is can be written in terms of $b_n$:
\begin{lemma}\label{lem_bsn_0}
	If $b_k=0 \; \forall k \neq s-1$ and 
	\begin{align*} 
	b_{s-1} &= -\frac{\lambda_s }{x_{1+2+\ldots + (s-1)}}
	\end{align*}
	then $b'_n=0 \; \forall n>1$.
\end{lemma}
\begin{proof}
	Follows from \cref{lem_bsn}: For $1<n<s-1$, setting $b_n=0$ amounts to $b'_n=0$. For $n=s-1$, the given choice of $b_{s-1}$ makes $b'_{s-1}$ vanish. The remaining $b_{n\geq s}$ vanish inductively since they  are proportional to $b_k$ for $k\leq n$ or $b'_{s-1}$, both of which are zero.
\end{proof}

The diffeomorphism coefficients $\{a_j\}_{j\geq 1}$ in \cref{def_diffeomorphism} need to be constants (with respect to momenta), hence, by \cref{thm_bn}, the $\{b_n\}_{n \geq 1}$ are constants, too. On the other hand, by \cref{lem_bsn_0}, the tree sum $b_{s-1}$ depends on the offshell variable $x_{1+2+\ldots + (s-1)}$ of the $s$-point-correlation function. Hence, this offshell variable needs to be constant, i.e. the attempted cancellation of quantum corrections only works for one specific offshell momentum. Call this momentum $p$, so
\begin{align*} 
x_{1+2+\ldots + (s-1)} &=:x_p\neq 0.
\end{align*}
For a fixed value of $x_p$ the quantity $b_{s-1}$ defined by \cref{lem_bsn_0} is a constant.

Next, we proceed to tree level amplitudes where more than one external edge is offshell, these are ${A'}^j_n$ as discussed in \cref{sec_offshelltree}. If $n<s$ they vanish with  the choice \cref{lem_bsn_0} since all $b_k=0$ for $k<s-1$. If $n \geq s$, ${A'}^j_n$ involves tree sums $b'_{s-1}(x_j)$ where $x_j$ is the offshell variable of one of the external edges. But by \cref{lem_bsn_0}, the  offshell variable needs to be the specific constant value $x_p$ in order for $b'_{s-1}$ to vanish. Thus, all offshell external edges in ${A'}^j_n$ need to have the offshell variable $x_p$ to make it vanish. Note that this does not fix a value of the momentum vector $p$ itself, only its magnitude. Kinematic configurations where all external offshell variables are either zero or $x_p$ are in principle possible in terms of overall momentum conservation.

\begin{theorem}\label{thm_adiabatic_phis}
	Let $s \geq 3$ be an integer and $\phi$  a scalar quantum field with mass $m \geq 0$ and   $\phi^s$-type interaction according to \cref{L_phis}. Let ${A'}^j_n$ be the connected tree level Feynman amplitude  with $n$ external edges of which $j>0$ are offshell. Assume all offshell external edges of ${A'}^j_n$ have the offshell variable $x_p$. Then the quantum field $\rho$   defined by 
	\begin{align*}
	\rho( x) &=\phi( x) - \frac{\lambda_s }{(s-1)! x_p}\phi^{s-1}( x) 
	\end{align*}
	has the same tree level amplitudes as a free quantum field with mass $m$, i.e. 
	\begin{align*} 
	{A'}^2_2 &= -ix_p\\
	{A'}^j_n &= 0 \ \forall n>2, j>0.
	\end{align*}	
\end{theorem}
\begin{proof}
	By \cref{sec_offshelltree} it is sufficient to have $b'_j=0 $ $\forall j$ to reach ${A'}^k_n=0$ $\forall n>2$. By \cref{lem_bsn_0} we know the values of $b_j$ to reach $b'_n=0$, by \cref{lem_inverse} these $b_j$ are the series coefficients of $\rho(x)$.	
\end{proof}

\begin{lemma}\label{lem_adiabatic_aj}
	The diffeomorphism $\phi(\rho)$ defined in \cref{thm_adiabatic_phis} has the  coefficients
	\begin{alignat*}{3}
	a_{s-2} &=   \frac{\lambda_s }{(s-1)!x_p} && \\
	a_{j \cdot (s-2)} &=  \left(a_{s-2}\right)^j\cdot F_j(s-1, 1), \qquad && j \in \mb N \\
	a_k &= 0 , \qquad &&k \notin \mb N \cdot (s-2).
	\end{alignat*}
	Here,  $ F_j(s-1, 1)$ are the Fuss-Catalan numbers \cref{fc}.
\end{lemma}
\begin{proof}By \cref{lem_bsn}, $b'_n=b_n$ for $n<s-1$. By \cref{thm_bn}, $b_n$ is a polynomial in $a_1, \ldots, a_{n-1}$. Hence if $a_k=0 \forall k<s-2$ then $b'_n=0 \forall n<s-1$ and $b_{s-1} = -(s-1)! a_{s-2}$. This fixes the claimed value of $a_{s-2} =:\alpha$. 
	
	Now define a function $C(t)$ such that
	\begin{align*} 
	\phi(\rho) = \rho \cdot C\left( \alpha \rho^{s-2} \right) .
	\end{align*}
	Insert this function for $\phi$ into  \cref{thm_adiabatic_phis} to obtain 
	\begin{align*} 
	\alpha \rho^{s-2} \cdot C^{s-1} \left( \alpha \rho^{s-2} \right) +1 &= C \left( \alpha \rho^{s-2} \right) .
	\end{align*}
	
	The generating function  of the Fuss-Catalan-Numbers \cref{fc},
	\begin{align}\label{fc_generating_series}
	C_{a,b}(t) &= \sum_{m=0}^\infty t^m F_m(a,b),
	\end{align}
	fulfills the functional equation 	\cite[theorem 2.1]{merlini_lagrange_2006}
	\begin{align}\label{fc_generating_functional}
	C_{a,b}(t) &= \left( tC_{a,b}^{\frac a b} (t)+1\right) ^b.
	\end{align}
	Thus $C(t)$ is the generating function of Fuss-Catalan-Numbers $\lbrace F_j(s-1, 1)\rbrace_j$, with argument $t=\alpha \rho^{s-2}$. Taking coefficients and identifying \cref{def_diffeomorphism} yields the values of $a_k$ for $k>s-2$.
\end{proof}

\subsection{Loop amplitudes}\label{subs_cance_loops}

To cancel the integrands of loop amplitudes, the idea is that the diffeomorphism of an interacting field should mimic the behaviour of the diffeomorphism of a free field in \cref{subs_free_loops}. There, \cref{lem_loop} was obtained because the tree sums $A^0_n=S_n$  of the free diffeomorphism vanish by \cref{lem_Sn_free} if the external momenta are onshell. We implicitly used the vanishing of tadpole graphs in \cref{lem_tadpole,lem_loop}.  

To cancel interaction vertices, by \cref{thm_adiabatic_phis} the situation is similar to the free case in terms of combinatorics, only that the tree sums ${A'}^j_n$ vanish if all their $j$ external offshell momenta have the fixed offshell variable $x_p\neq 0$ instead of being onshell. On the other hand, the tree sums ${A'}^0_n$ do not vanish, hence a graph may have arbitrarily many such meta-vertices as internal vertices. Assume that  there are $j>0$ external edges with offshell variable $x_p\neq 0$ and all remaining external edges are onshell. 
Let $e$ be one of the external edges with offshell variable $x_p$. Then $e$ is connected to a meta vertex ${A'}^j_n$ with $j>0$ and $n>2$ which vanishes due to \cref{thm_adiabatic_phis}. By \cref{subs_cuts}, the Feynman integrand of the sum of all amplitudes with given loop number is proporitonal to ${A'}^j_n$, hence it is zero. Consequently, as soon as there is at least one external edge with offshell variable $x_p$ and all remaining edges are either also $x_p$ or onshell, the integrand vanishes.  This implies \cref{thm_adiabatic_phis} holds not only for tree level amplitudes, but also for the integrands of loop amplitudes if tadpoles vanish.

\section{Scalar field with arbitrary propagator}\label{sec_arbitrary}

\subsection{Generalized free Lagrangian density}

By \cref{thm_bn} the tree level-amplitudes of the diffeomorphism are independent of masses and momenta, which suggests that the actual form of the propagator is unimportant for the result.  In this section we show that  this is the case. We restrict to a free scalar field, that is, the Lagrangian is quadratic in the field variable $\phi$ and its derivatives. Lorentz invariance dictates that any derivative $\partial_\mu \phi$ be contracted with $\partial^\mu \phi$. By partial integration, all derivatives can be concentrated on one of the two field variables up to a vanishing total derivative. We assume a  Lorentz invariant free scalar Lagrangian density of the form 
\begin{align}\label{diff_skalar_beliebig_L}
\mc L_\phi &= \frac 12 \phi(x) \hat X \phi(x)
\end{align}
where $\hat X$ is a linear Lorentz invariant differential operator
\begin{align*} 
\hat X &= \beta_0 + \beta_1 \left(-\partial_\mu \partial^\mu\right)  + \beta_2  \left( -\partial_\mu \partial^\mu  \right) ^2 + \ldots = \sum_{k=0}^\infty \beta_k \left( -\partial_\mu \partial^\mu  \right) ^k.
\end{align*}
$\left \lbrace \beta_j \right \rbrace _j$ are spacetime-independent constants. The ordinary free field \cref{Lfree} amounts to $\beta_0=-m^2, \beta_1 = 1, \beta_{j>1}=0$. Recently, Lagrangians involving higher derivatives have for example been studied in the context of fixed points in more than four spacetime dimensions \cite{gracey_higher_2017} or conformal field theory \cite{brust_free_2017}.
We do not discuss here for which values of $\beta_j$  the Lagrangian \cref{diff_skalar_beliebig_L} describes a physically meaningful, e.g. unitary or renormalizable, theory, but are just interested in its formal behaviour under field diffeomorphism.  

The classical field defined by \cref{diff_skalar_beliebig_L} has the equation of motion
\begin{align}\label{diff_skalar_beliebig_eom}
\hat X \phi(x) &= 0
\end{align}
which is linear as desired for a free field. In momentum space, the operator $\hat X$ acts as multiplication,
\begin{align*}
\hat X \phi(  x) &= \int \frac{\d^4  k }{(2 \pi)^4 }\underbrace{\left( \beta_0 + \beta_1 k_\mu k^\mu  +\beta_2 k_\mu k^\mu k_\nu k^\nu + \ldots  \right)}_{=:X_k} \phi( k) e^{-ik x}.
\end{align*}

The propagator of the field $\phi$ is defined as the inverse of the field differential operator \cref{diff_skalar_beliebig_eom} (i.e. as a Green's function) in momentum space,
\begin{align}\label{diff_skalar_beliebig_propagator}
 \Gamma_{2, \text{free}}^{-1} (k) = \langle \phi(k)  \phi(-k)\rangle &= \frac{i}{X_k}.
\end{align}
As long as $\beta_1 \neq 0$, this propagator can also be obtained  by the standard perturbative calculation of propagators: Use the offshell variable \cref{def_offshellvariable} to write $X_k = x_k + X'_k$ where $X'_k$ contains all higher derivative terms and is to be treated as a perturbation. If  $\frac{i}{x_k}$ is identified as the propagator of $\phi$, the remainder $X'_k$  gives rise to a 2-valent vertex with Feynman rule $iX'_k$. Summing arbitrary many of these vertices, connected with a propagator $\frac{i}{x_k}$, one obtains a geometric series with the sum $\frac{i}{x_k + X'_k } = \frac{i}{X_k}$.  

Note that in this setting, the definition of ``offshellness'' is according to this propagator, i.e. the momentum $k$ is offshell iff $X_k\neq 0$, regardless if $x_k=0$ or not. We do not care here if the pole $X_k=0$ of the propagator \cref{diff_skalar_beliebig_propagator} corresponds to a one-particle state (i.e. is a simple pole with unit residue) as in the standard interpretation of perturbative quantum field theory.

\subsection{Diffeomorphism Feynman rules}

Applying the diffeomorphism \cref{def_diffeomorphism} to \cref{diff_skalar_beliebig_L} gives rise to the Lagrangian
\begin{align}\label{diff_skalar_beliebig_L_rho}
\mc L_\rho &= \frac 12 \sum_{j=0}^\infty \sum_{k=0}^\infty a_j a_k \rho^{j+1} \hat X \rho^{k+1}.
\end{align}
Since $a_0=1$, the propagator of $\rho$ coincides with that of $\phi$,  
\begin{align*}
\left \langle \rho( p) \rho(-p)\right \rangle  &= \frac{i}{X_p}.
\end{align*}

To derive the vertex Feynman rules of $\rho$, consider the Fourier transform of the action of $\hat X$ on a product of fields. 
\begin{align*}
\hat X \phi^n (  x) &=  \idotsint \frac{\d^4   k_1}{(2\pi)^4} \cdots \frac{\d^4  k_n}{(2\pi)^4}  \phi( k_1) \cdots \phi( k_n)  X_{1+\ldots + n} \cdot  e^{-i   k_1 x} \cdots e^{-i  k_n   x}  .
\end{align*}
We have defined $X_{1+\ldots +n}$ analogous to \cref{def_offshellvariable}, e.g.
\begin{align*} 
X_{1+2} &= \beta_0 + \beta_1 \left( k_1 + k_2 \right) _\mu \left( k_1 + k_2 \right) ^\mu + \beta_2\left( \left( k_1 + k_2 \right) _\mu \left( k_1 + k_2 \right) ^\mu \right)^2+ \ldots. 
\end{align*}
The $n$-valent diffeomorphism vertex arises from terms $\phi^j \hat X \phi^k$ in \cref{diff_skalar_beliebig_L_rho} where $j+k=n$. Summing over all permutations of $n$ edges, this amounts to $j!$ times a sum over all ways to choose $k$ out of $n$, each of them multiplied by $k!$ because $X_{1+\ldots + k}$ is symmetric under permutation of its indices. Hence the vertex Feynman rule is 
\begin{align}\label{diff_skalar_beliebig_vn}
i v_n &= i \frac 12 \sum_{k=1}^{n-1} a_{n-k-1} a_{k-1} (n-k)! k! \sum_{P\in Q^{(n,k)}} X_{P }.
\end{align}
The latter sum consists of all $\abs{ Q^{(n,k)}  } = \binom n k$ possibilities to choose $k$ out of $n$ external edges without distinguishing the order,
\begin{align}\label{Qnk} 
Q^{(n,k)} :=\left \lbrace \left \lbrace  j_1, \ldots, j_k \right \rbrace   \subseteq  \left \lbrace 1, \ldots, n \right \rbrace \right \rbrace   .
\end{align}
Of course, these are the same partitions which make up elementary symmetric polynomials, i.e. the elementary symmetric polynomial of degree $k$ in variables $\left \lbrace z_1, \ldots, zn \right \rbrace $ is $e_k(\left \lbrace z_1, \ldots, z_n \right \rbrace  )= \sum_{P\in Q^{(n,k)}} \prod_{z\in P} z $ .
\begin{example}
	For $n=5, k=3$ there are ten summands
	\begin{align*} 
	\sum_{P\in Q^{(5,3)}} X_{P} &= X_{1+2+3} + X_{1+2+4} + X_{1+2+5}+ X_{1+3+4} + X_{1+3+5} \\
	&\qquad + X_{1+4+5} + X_{2+3+4} + X_{2+3+5} + X_{2+4+5} + X_{3+4+5}.
	\end{align*}
	The indices correspond to the summands of the elementary symmetric polynomial
	\begin{align*} 
	e_3(\left \lbrace z_1, \ldots, z_5 \right \rbrace  ) &= z_1 z_2 z_3 + z_1 z_2 z_4 + \text{ (6 more) } +  z_2 z_4 z_5+ z_3 z_4 z_5.
	\end{align*}
\end{example}

Momentum is conserved at each vertex, $ p_1+ \ldots +  p_n =  0$, hence only the elements of $Q^{(n, k)}$ for $k \leq \frac n 2$ are linearly independent. For $k>\frac n 2$, each $X_P$ coincides with a term which was already present, this mechanism effectively cancels the prefactor $\frac 1 2$ in the Feynman rule \cref{diff_skalar_beliebig_vn}.  If $n$ is even, the terms in  $Q^{(n, \frac n 2)}$ coincide pairwise, again cancelling the $\frac 1 2$.

\begin{example}Using $X_{1+2} \equiv X_3$ etc., the 3-valent vertex \cref{diff_skalar_beliebig_vn} is
	\begin{align*}
	iv_3 &= i \frac 12 a_1 2! 1!\left( X_1+X_2 + X_3 \right) + i   \frac 12 1! 2! \left( X_{1+2} + X_{1+3}+ X_{2+3} \right)  \\
	&= 2ia_1 \left( X_1 + X_2 + X_3 \right) . 
	\end{align*}
	This result coincides with \cref{free_vertex_picture}, but all higher vertices do not, e.g.
	\begin{align*}
	iv_4 	&= 6 ia_2 \left( X_1 + X_2 + X_3 + X_4 \right) + 4i a_1^2 \left( X_{1+2} + X_{1+3} + X_{2+3} \right) . 
	\end{align*}
	This is because for higher than second powers, the sum of momenta does not decompose, i.e. $X_{1+2}$ cannot generally be expressed as $X_1, X_2$ and $m^2$.
\end{example}

\subsection{Cancellation of internal edges}

The vertex Feynman rule \cref{diff_skalar_beliebig_vn} involves summands $X_j$ which directly cancel an adjacent propagator $\frac{i}{X_j}$, but also summands $X_{j+k+\ldots}$. These summands correspond to a partition of the external edges into two subsets and it turns out they cancel the corresponding contribution of two vertices connected with an internal edge.

\begin{lemma}\label{diff_skalar_beliebig_lem_intern}
	Let $e$ be an internal edge in a tree sum $A^n_n$ from \cref{def_Ajn}. Then there is no summand proportional to $X_e$ in $A^n_n$. 
\end{lemma}
\begin{proof}
	Consider one arbitrary tree  $A\in A^n_n$. 
	Since $e$ is internal, it is connected to two vertices. Let their valence be $j$ resp. $k$, then $e$ together with the two vertices form a subtree $T$ which has $j+k-2\leq n$ external edges, some (or all) of which may be external edges of $A^n_n$. Its Feynman rule is
	\begin{align*}
 	T&=	i v_j \cdot \frac{i}{X_e}\cdot i v_k,
	\end{align*}
	We identify the external edges of $T$ with numbers   $1, \ldots, j+k-2$ and consider terms proportional to $X_e$ in $T$.  This requires both vertices $v_j, v_k$ to cancel the edge $e$ connecting them. By  \cref{diff_skalar_beliebig_vn} this summand has the Feynman amplitude 
	\begin{align*}
	i    a_{k-2}  (k-1)!  X_e   \frac{i}{X_e}  i a_{j-2} (j-1)! X_e &= -i a_{j-2} a_{k-2} (j-1)! (k-1)! X_e.
	\end{align*}
	
	The momentum running through $e$ corresponds to one way to partition the external momenta  $\left \lbrace 1, \ldots, j+k-2 \right \rbrace $ of $T$ into two sets, namely into  $(k-1)$ and $(j-1)$ elements.
	
	The tree sum $A^n_n$ also contains a tree $A'$ where  $T$ is replaced by a vertex $iv_{j+k-2}$ at the same positon, see \cref{diff_skalar_beliebig_T}. By \cref{diff_skalar_beliebig_vn} this vertex is made of summands proportional to $X_P$ where $P$ is any way to partition the $j+k-2$ external edges into two sets. Especially, precisely two such partitions resembles the momentum running through $e$ such that $X_P=X_e$. They contribute to $iv_{j+k-2}$ with an amplitude 
	\begin{align*}
	i \frac 12  a_{j-2}a_{k-2} (j-1)! (k-1)! X_e+ i \frac 12 a_{k-2} a_{j-2} (k-1)! (j-1)!X_e &= i a_{j-2} a_{k-2} (j-1)! (k-1)! X_e.
	\end{align*}
	This contribution cancels the contribution from $T$. 
	
	 We did not impose any conditions on $T$, therefore the cancellation works in any case: Whenever there is an internal edge, the sum $A^n_n$ contains the same term with a vertex at the same position which cancels the contribution proportional to $X_e$.
\end{proof}

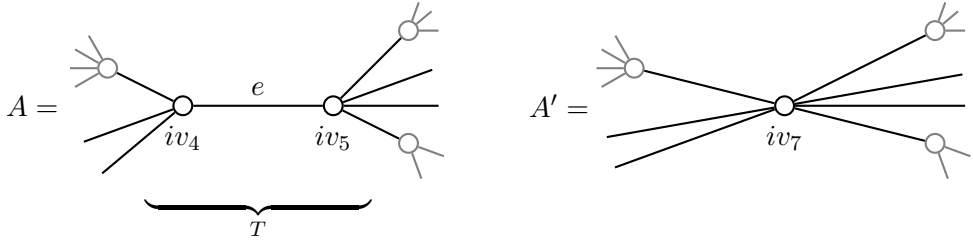
\begin{figure}[htbp]
	\centering
	\begin{tikzpicture}

	\node at (-1,0) {$A=$};

	\node [diffVertex,label=below:$iv_4$] (v1) at (1,0){};
	\node [diffVertex,label=below:$iv_5$] (v2) at (3,0){};
	\node [diffVertexGray] (v3) at (4, 1){};
	\node [diffVertexGray] (v4) at (4, -.5){};
	\node [diffVertexGray] (v5) at (0, .5){};
	 
	\draw [thick] (v1) -- node[above]{$ e$} (v2);
	\draw [thick] (v1) -- (v5);
	\draw [thick] (v3) -- (v2);
	\draw [thick] (v4) -- (v2);

	\draw [-,thick] (v1) -- +(200:1.4);
	\draw [-,thick] (v1) -- +(220:1.4);
	\draw [-,thick] (v2) -- +(0:1.4); 
	\draw [-,thick] (v2) -- +(20:1.4); 
	\draw [gray,thick] (v3) -- +(70:.4);
	\draw [gray,thick] (v3) -- +(40:.4);
	\draw [gray,thick] (v3) -- +(0:.4);
	\draw [gray,thick] (v4) -- +(290:.5);
	\draw [gray,thick] (v4) -- +(340:.5);
	\draw [gray ,thick] (v5) -- +(120:.5);
	\draw [gray ,thick] (v5) -- +(150:.5);
	\draw [gray,thick] (v5) -- +(180:.5);
	\draw [gray,thick] (v5) -- +(210:.5);

	\node at (2, -1.5) {$\underbrace{\hspace{3cm} }_{T}$};

		\node at (6,0) {$A'=$};

	\node [diffVertex,label=below:$iv_7$] (v1) at (9,0){};
	\node [diffVertexGray] (v3) at (11, 1){};
	\node [diffVertexGray] (v4) at (11, -.5){};
	\node [diffVertexGray] (v5) at (7, .5){};
	 
	\draw [thick] (v1) -- (v5);
	\draw [thick] (v3) -- (v1);
	\draw [thick] (v4) -- (v1);

	\draw [-,thick] (v1) -- +(190:2.4);
	\draw [-,thick] (v1) -- +(200:2.4);
	\draw [-,thick] (v1) -- +(0:2.4); 
	\draw [-,thick] (v1) -- +(10:2.4); 
	\draw [gray,thick] (v3) -- +(70:.4);
	\draw [gray,thick] (v3) -- +(40:.4);
	\draw [gray,thick] (v3) -- +(0:.4);
	\draw [gray,thick] (v4) -- +(290:.5);
	\draw [gray,thick] (v4) -- +(340:.5);
	\draw [gray ,thick] (v5) -- +(120:.5);
	\draw [gray ,thick] (v5) -- +(150:.5);
	\draw [gray,thick] (v5) -- +(180:.5);
	\draw [gray,thick] (v5) -- +(210:.5);

	\end{tikzpicture}
	\caption{ Illustration of contributions to a tree sum for $j=4, k=5$. $A$ contains a subtree $T$ consisting of one internal edge $e$ connected to two vertices whereas in $A'$ that subtree is replaced by a single vertex.}
	\label{diff_skalar_beliebig_T}
\end{figure}

The vertex Feynman rules \cref{diff_skalar_beliebig_vn} allow for a transparent discussion of the cancellation of internal edges in tree sums. Therefore we believe that they are a more intuitive form than \cref{diff_vn} even in the case of an ordinary quadratic propagator. Especially, it is nice to see how the Lagrangian being quadratic in $\phi$ (i.e. free) translates to the Feynman rules containing partitions of edges into two subsets. The special discussion of possible constant terms in the Feynman rules is also no longer necessary: The vertex \cref{diff_skalar_beliebig_vn} contains only terms proportional to some $X_P$, the constants $m^2$ in \cref{diff_vn} were mere artifacts of decomposing $x_{i+j+\ldots}$.

\begin{theorem}\label{thm_A1n}
	For any $n>2$ there is a  $C_n (a_1, a_2, \ldots)$ which does not depend on kinematics such that
	\begin{align*}
	A_n^1 = -iC_n \cdot \left( X_1 + \ldots + X_n \right) .
	\end{align*}
\end{theorem}
\begin{proof}
	By the Euler characteristic \cref{euler_graph}, a tree $A\in A_n^1$ with $\abs{V_A}$ vertices contains $\abs{V_A}-1$ internal edges. By \cref{diff_skalar_beliebig_vn} each vertex is proportional to some $X$. Each internal edge has an propagator proportional to $\frac{1}{X}$ (these are not the same $X_k$, we just count powers of $X$). Consequently, $A_n^1$ is in total proportional to $X^1$. 
	
	By \cref{diff_skalar_beliebig_lem_intern} there is no summand $\propto X_e$ with any internal edge $e$ of $A^1_n$. Also, each vertex \cref{diff_skalar_beliebig_vn} is proportional to $X^1$ (and not e.g. quadratic in $X$) and each external edge is connected to only one vertex in any $A\in A^1_n$. Therefore there can be no summand proportional to $X_j^{k>1}$ for an external edge $j$. Further, there are no factors $X^{-1}_j$ for external edges $j$ since their propagators are not included in $A^1_n$. Consequently, an individual tree $A\in A^1_n$ must be of the form $C_n \frac{X_j X_k \cdots }{X_e X_f \cdots}$ where $j\neq  k\neq  \ldots$ are external momenta and $e\neq  f\neq  \ldots$ internal edges (i.e. partitions of the external momenta) and there is one more factor in the numerator than in the denominator. But in $A^1_n$, only one of the $n$ external edges is offshell, consequently there can be at most one non-zero factor in the numerator. The only non-zero term contributing to $A$ is of the form 	
	  $X_j \cdot c_n$ were $c_n$ does not depend on any $X$. Since $X$ represent the only possible dependence on kinematics in the Feynman rules \cref{diff_skalar_beliebig_vn}, the constant $c_n$ must be independent of kinematics. 
	  
	  The tree sum $A^1_n$ is symmetric in external momenta, hence it must be proportional to $\sum_{j=1}^n X_j $ where again the proportionality constant does not depend on any $X_P$ but only on the diffeomorphism-parameters $\left \lbrace a_1, \ldots,  \right \rbrace  $.

\end{proof}
Consistency of \cref{thm_A1n,def_bn} in the case $X_k \rightarrow x_k$ requires that $C_n= b_{n-1}$ are the known coefficients from \cref{thm_bn} resp \cref{lem_inverse2}.

\begin{lemma}
	The diffeomorphism does not alter the $S$-matrix and $S_n=A_n^0=0 \ \forall n>2$ .
\end{lemma}
\begin{proof}
 This is analogous to \cref{lem_Sn_free}. The vanishing  follows immediately from \cref{thm_A1n} since there is no constant term which survives if all $X_j=0$. For loop amplitudes, the discussion from \cref{subs_free_loops} is equally valid here where just all $x_j$ are to be replaced by $X_j$. 
\end{proof}

\subsection{Recursive definition of tree sums}

 The proof of \cref{thm_bn} in \cite{KY17} relied on a recursive definition of $b_n$. Such a definition is possible in the arbitrary-propagator-case as well: $b_n$ has one upper vertex with valence $k+1$ connected to $k$ smaller tree sums. One has to sum all partitions of the $n$ lower edges into $k$ subsets and also all numbers $k$ of subtrees. Additionally, the offshell propagator $\frac{i}{X_{1+\ldots + n}}$ has to be included and all the lower propagators $\left \lbrace X_l \right \rbrace _{l\in \left \lbrace 1, \ldots, n \right \rbrace  }$ are set to zero. This yields

\begin{align*} 
b_n &= - \frac{1}{X_{1+\ldots + n}} \cdot  \sum_{k>1} \sum_{\left \lbrace P_j \right \rbrace  \in P^{(n,k)}} b_{\abs{P_1}} \cdots b_{\abs{P_k}} \cdot v_{k+1} \Big|_{\left \lbrace X_l \right \rbrace  =0}\\
&= -  \frac{1}{2X_{1+\ldots + n}} \sum_{k=2}^n\sum_{\left \lbrace P_j \right \rbrace  \in P^{(n,k)}} b_{\abs{P_1}} \cdots b_{\abs{P_k}}\cdot \sum_{j=1}^{k} a_{k-j} a_{j-1} (k+1-j)! j! \sum_{Q\in Q^{(k+1,j)}} X_{Q}   \Big|_{\left \lbrace X_l \right \rbrace  =0}.
\end{align*}
Here, $P^{(n,k)}$ is the set of all possible partitions of $\left \lbrace 1, \ldots, n \right \rbrace $ into $k$ nonempty disjoint subsets
\begin{align*} 
P^{(n,k)} &= \left \lbrace \left \lbrace P_j \right \rbrace _j:  P_1 \sqcup \ldots \sqcup P_k = \left \lbrace 1, \ldots, n \right \rbrace  , \abs{P_j} \geq 1\right \rbrace  
\end{align*}
and $Q^{(n,k)}$ from \cref{Qnk} is the set of all $k$-element subsets of $\left \lbrace 1, \ldots, n \right \rbrace$. Note that in our case, $Q^{(k+1, j)}$ are not necessarily subsets of external edges but of the top edges of the lower $b_j$. These in turn are given by the partition $P_j$ and one of the external edges of $v_{k+1}$ is the top edge which carries momentum $(1+\ldots + n )$, hence
\begin{align*} 
Q^{(k+1,j)}: =\left \lbrace  \left \lbrace P_{i_1}, \ldots, P_{i_{j}} \right \rbrace \subseteq \left \lbrace P_1, \ldots, P_k, (1+\ldots + n )\right \rbrace   \right \rbrace  .
\end{align*}

\begin{example} 
	The first nontrivial tree sum $b_2$ only involves the summand  $k=2$. Consequently, the only possible partition of the two lower edges is $1+1$ and 
	\begin{align*} 
	b_2 &= -\frac 12 \frac{1}{X_{1+2}} \sum_{j=1}^{2} a_{2-j} a_{j-1} (2+1-j)! j! \sum_{Q\in Q^{(3,j)}} X_{Q}   \Big|_{X_1=0=X_2} \\
	&= -\frac{1}{2X_{1+2}} \left( a_1 a_0 2! 1! \left( X_1+X_2  + X_{1+2} \right) + a_0 a_1 1! 2! \left( X_{1+2} + X_{1} + X_2\right)  \right)  \Big|_{X_1=0=X_2}\\
	&= -2 a_1.
	\end{align*}
	This coincides with \cref{ex_bn}. For the next tree $b_3$, two values of $k$ contribute:
	\begin{align*} 
	b_3 &= -\frac{1}{2 X_{1+2+3}} \left( 3  b_2  \sum_{j=1}^2 a_{2-j} a_{j-1} (3-j)! j!\sum_{Q\in Q^{(3,j)}} X_{Q}   +   \sum_{j=1}^3 a_{3-j} a_{j-1} (4-j)! j! \sum_{Q\in Q^{(4,j)}} X_{Q}    \right)  .
	\end{align*}
	Computing the sums and inserting $\left \lbrace X_l \right \rbrace =0$ one again obtains the same result as in \cref{ex_bn},
	\begin{align*} 
	b_3 &= -\frac{1}{2 X_{1+2+3}} \left( -8 a_1^2 \left( X_{1+2} + X_{2+3} + X_{1+3} + 3X_{1+2+3} \right)  \right. \\
	&\qquad \left. + 12 a_2 \left( X_{1+2+3} \right) + 8a_1^2 \left( X_{1+2} + X_{2+3} + X_{1+3} \right)   \right) \\
	&=12 a_1^2-6a_2.
	\end{align*}
\end{example}

\subsection{Interacting theory}

If instead of \cref{diff_skalar_beliebig_L} the underlying Lagrangian density is
\begin{align} \label{diff_skalar_beliebig_L_interacting}
\mc L_\phi &= \frac 12 \phi(x) \hat X \phi (x) -\frac{\lambda_s }{s!}\phi^s( x),
\end{align}
additionally the vertices $-iw^{(s)}_n$ given by \cref{diff_arb_wns} are present. But since by \cref{thm_A1n} each tree sum cancels precisely one adjacent propagator $\frac{i}{X_j}$, the argument from \cref{subs_cancellation_higher} works even for non-standard propagators. Especially, \cref{thm_phis_Sn} still holds.

\section{Non-local field transformations}\label{subs_nonlocal}

We noted below \cref{arb_generalvertex} that even the most general vertex which can arise by a local diffeomorphism \cref{def_diffeomorphism} of a scalar interacting quantum field theory \cref{L_arb} has a very restricted momentum dependence which is essentially determined by the propagator resp. offshell variable \cref{def_offshellvariable} of adjacent edges. On the other hand, in \cref{sec_arbitrary} we saw that the combinatorics of the diffeomorphism are independent of the momentum dependence of the propagator. In this section we include derivatives into the transformation $\phi \mapsto \rho$ in order to change the momentum dependence of the propagator of a free field. 

\subsection{Arbitrary propagator from field transformation}

We first consider a transformation involving derivatives, but not a local diffeomorphism, i.e.
\begin{align}\label{nonlocal_transform_only} 
\phi (x) &= \sum_{k=0}^\infty \alpha_k \left( -\partial_\mu \partial^\mu \right) ^k \rho (x).
\end{align}
Applied to the standard free Lagrangian \cref{Lfree}, this yields (after $2k+1$ partial integrations)
\begin{align}\label{L_nonlocal_only} 
\mc L_\rho &= \frac 12  \rho (x)  \sum_{n=0}^\infty \sum_{m=0}^n \alpha_{n-k} \alpha_k \left( -\partial_\mu \partial^\mu -m^2  \right)  \left( -\partial_\mu \partial^\mu \right) ^{n} \rho (x).
\end{align}
This equals the Lagrangian \cref{diff_skalar_beliebig_L} with the operator
\begin{align*} 
\hat X &= \sum_{n=0}^\infty \sum_{k=0}^n \alpha_{n-k} \alpha_k \left( -\partial_\mu \partial^\mu -m^2  \right)  \left( -\partial_\mu \partial^\mu \right) ^{n}= \sum_{k=0}^\infty \beta_k \left( -\partial_\mu \partial^\mu  \right) ^k
\end{align*}
where
\begin{align}\label{nonlocal_betan} 
\beta_n &= \sum_{k=0}^{n-1} \alpha_{n-1-k} \alpha_k -m^2 \sum_{k=0}^n \alpha_{n-k} \alpha_k .
\end{align}
By this, the transformed field $\rho$ can be identified with the underlying field $\phi$ in the free field case in \cref{sec_arbitrary}.
 
\subsection{Local diffeomorphism and non-local transformation}\label{subs_local_nonlocal}

\Cref{L_nonlocal_only} suggests that with \cref{nonlocal_betan} the non-locally transformed field $\rho$  from \cref{nonlocal_transform_only} takes the role of the underlying field $\phi$ in the local diffeomorphism. That is, the local diffeomorphism \cref{def_diffeomorphism} replaces $\rho$ in \cref{nonlocal_transform_only}.

\begin{definition}\label{def_nonlocal_transformation}
	In a combined transformation, the field $\phi(x)$ is replaced by 
	\begin{align*} 
	\phi(x) &= \sum_{k=0}^\infty \alpha_k \left( -\partial_\mu \partial^\mu  \right) ^k  \left( \sum_{j=0}^\infty a_j \rho^{j+1} (x) \right) 
	\end{align*}
	where both $\left \lbrace \alpha_k \right \rbrace _{k \geq 0}$ and $\left \lbrace a_j \right \rbrace _{j\geq 0}$ are constants and $a_0=1$.
\end{definition}
Indeed, an explicit calculation shows that it is not possible to have $\alpha_k$ depend on $j$, i.e. to transform different monomials with different derivative operators. If that were the case, the cancellation property of the local diffeomorphism would break down as it would fail to produce vertices proportional to an inverse propagator.

The transformation \cref{def_nonlocal_transformation} applied to a free Lagrangian \cref{Lfree}   gives rise to a field $\rho$ which has propagator \cref{diff_skalar_beliebig_propagator}
\begin{align}\label{nonlocal_propagator} 
\langle \rho(k) \rho(-k) \rangle &= \frac{i}{X_k} = \frac{i}{\beta_0 + \beta_1 k^2 + \beta_2 (k^2)^2 + \ldots }
\end{align}
with $\beta_n$ from \cref{nonlocal_betan} and $(n\geq 3)$-valent interaction vertices \cref{diff_skalar_beliebig_vn},
\begin{align*}
i v_n &= i \frac 12 \sum_{k=1}^{n-1} a_{n-k-1} a_{k-1} (n-k)! k! \sum_{P\in Q^{(n,k)}} X_{P }.
\end{align*}
By \cref{thm_A1n}, all $(n>2)$-valent tree sums vanish in the onshell limit $X_j=0\ \forall j\in \left \lbrace 1, \ldots, n \right \rbrace $, thus $\rho$ is a free theory. Note that on the other hand  its 2-point function stays  \cref{nonlocal_propagator}, so the derivative part of the non-local transform \cref{def_nonlocal_transformation} has a potential influence on observables. 

In this paper, we do not consider  the effect of \cref{def_nonlocal_transformation} on an underlying interacting Lagrangian.

\section{Conclusion}

We have continued the perturbative examination of diffeomorphisms of quantum field started in \cite{velenich,KY17} and reached the following results:
\begin{enumerate}
	\item Even if the underlying Lagrangian is not a free theory, but contains self-interaction, still a local diffeomorphism does not contribute to the $S$-matrix (\cref{thm_Alm}) provided tadpole graphs vanish.  Thus, scalar quantum field theories related by diffeomorphism can be regarded physically equivalent.
		\item The invariance of the $S$-matrix under local field diffeomorphisms does not require a special form of the propagator (\cref{thm_A1n}). 
	\item The coefficients of the diffeomorphism of an interacting theory can be chosen such that for an offshell amplitude, the interaction is cancelled at tree level and the integrands of loop amplitudes are zero at a specific external momentum up to tadpole graphs (\cref{thm_adiabatic_phis}) .
	\item In the case of an underlying free Lagrangian, it is possible to extend the local diffeomorphism to a transformation involving  derivatives in a way that preserves the combinatoric structure of the local diffeomorphism (\cref{subs_local_nonlocal}).
\end{enumerate}
Apart from the natural task to extend the study to charged or non-scalar fields, open questions  to be covered in future work are:
\begin{enumerate}
	\item There should be a simple combinatoric explaination why the tree sums $b_n$ from \cref{def_bn} are the coefficients of the inverse diffeomorphism (\cref{lem_inverse,lem_inverse2}), perhaps by a Legendre transform \cite{jackson_robust_2017,weinberg2}.
	\item By power counting, the local diffeomorphism $\rho$ is non-renormalizable, whereas by the above theorems we (in principle) know its $S$-matrix values, which are finite if the underlying theory is renormalizable. At which point do assumptions about renormalizability enter and is it possible to carry out a renormalization procedure for the field $\rho$ directly, i.e. without undoing the diffeomorphism?
\end{enumerate}

.

 
\printbibliography

\end{document}